\documentclass[a4paper,USenglish,cleveref,nameinlink,thm-restate,numberwithinsect]{lipics-v2021}
\pdfoutput=1

\usepackage{amsmath,amsthm,amssymb}
\usepackage{thmtools}
\usepackage{mathtools}
\usepackage{placeins}

\usepackage{tikz}
\usetikzlibrary{patterns}
\usetikzlibrary{calc} % for coordinates calculations
\usetikzlibrary{decorations.pathreplacing}
\usetikzlibrary{math} %for tikzmath and variables
\usetikzlibrary{arrows,arrows.meta}
\usetikzlibrary{positioning}
\usetikzlibrary{fit}

    \gdef\affinemakeabove{%
    \xdef\affineedgestyle{affineedgeabove} 
    \xdef\affinelabelstyle{affinelabelabove}
    }
    
    \gdef\affinemakebelow{%
    \xdef\affineedgestyle{affineedgebelow}  
    \xdef\affinelabelstyle{affinelabelbelow}
    }
    
    \affinemakeabove

    \newcommand{\affinesetup}[2][(0,0)]{%
    \gdef\makeaffinestartnode{\node[inner sep=0] (affinestartnode) at #1 {};}
    \foreach[count=\i from 0, count=\j from 1] 
		\len/\lower/\upper/\color/\label in {#2} {
		\expandafter\xdef\csname len\i\endcsname{\len}
		\expandafter\xdef\csname lower\i\endcsname{\lower}
		\expandafter\xdef\csname upper\i\endcsname{\upper}
		\expandafter\xdef\csname color\i\endcsname{\color}
		\expandafter\xdef\csname label\i\endcsname{{\label}}
		\expandafter\xdef\csname next\i\endcsname{\j}
		\xdef\maxi{\i}
		\xdef\maxiplus{\j}
	}    
    }
	
	\gdef\drawapsfromi#1{
		\ifnum#1>\maxi
			\affineenddraw
		\else	
			\node (current#1) at (affinestartnode) {};
			
			\edef\curlen{\csname len#1\endcsname}
			\edef\curlower{\csname lower#1\endcsname}
			\edef\curupper{\csname upper#1\endcsname}
			\edef\curcolor{\csname color#1\endcsname}
			\edef\curlabel{\csname label#1\endcsname}
			\edef\nextstep{\csname next#1\endcsname}

			\foreach \i in {1,...,\curupper} {
				\draw[\curcolor] (current#1.center) to[\affineedgestyle] 
					node[pos=1] (current#1) {}
					node[pos=0, inner sep=0] (affinemax#1-left) {} 
					node[pos=0.5, inner sep=0] (affinemax#1-center) {} 
					node[pos=1, inner sep=0] (affinemax#1-right) {} 
					node[pos=1, inner sep=0] (affinemax#1) {} 
					node[pos=1, inner sep=0] (affinemax) {} 
					node[\affinelabelstyle] {\curlabel}
					++(\curlen, 0);
				\ifnum\i<\curlower\else
					\node (affinestartnode) at (current#1) {};
					\expandafter\drawapsfromi\nextstep
				\fi
			}
		\fi	
	}
	
	\newcommand\affinedrawing[1][]{
		\makeaffinestartnode
		\gdef\affineenddraw{#1}
		\tikzset{affineedgeabove/.style={in=90,out=90}}
    	\tikzset{affineedgebelow/.style={in=270,out=270}}
    	\tikzset{affinelabelabove/.style={midway,above}}
    	\tikzset{affinelabelbelow/.style={midway,below}}
		\edef\curlen{\csname len0\endcsname}
		\edef\curcolor{\csname color0\endcsname}
		\edef\curlabel{\csname label0\endcsname}
		\draw[\curcolor] (affinestartnode.center) to
			node[pos=1] (affinestartnode) {}
			node[pos=0, inner sep=0] (affinemin) {}
			node[pos=0, inner sep=0] (affineroot-left) {}
			node[pos=0.5, inner sep=0] (affineroot-center) {}
			node[pos=1, inner sep=0] (affineroot-right) {}
			node[\affinelabelstyle] {\curlabel}
			++(\curlen, 0);
	\drawapsfromi{1}
	\affinemakeabove}

	\newcommand{\affinelower}[1]{\expandafter\relax\csname lower#1\endcsname}
	\newcommand{\affineupper}[1]{\expandafter\relax\csname upper#1\endcsname}

\newcommand{\strongconst}{5}
\newcommand{\strongconstverb}{five}

\pgfmathparse{int(\strongconst-1)}
\edef\strongconstminus{\pgfmathresult}
\pgfmathparse{int(\strongconst+1)}
\edef\strongconstplus{\pgfmathresult}

\newcommand{\textop}[1]{\textnormal{\textsf{#1}}}

\newcommand{\dd}{.\,.}
\newcommand{\Oh}{{\mathcal O}}
\newcommand{\rot}{\textop{rot}}
\newcommand{\pal}{\ensuremath{\textop{PAL}}}
\newcommand{\palk}{\ensuremath{\pal^k}}

\newcommand{\palpref}{\ensuremath{\mathsf{PalPref}}}
\newcommand{\palprefs}{\ensuremath{\palpref^s}}
\newcommand{\Aa}{\mathcal{A}}
\newcommand{\Bb}{\mathcal{B}}

\newcommand{\Cc}{\mathcal{C}}

\newcommand{\Pp}{\mathcal{P}}

\newcommand{\Ss}{\mathcal{S}}
\newcommand{\Ww}{\mathcal{W}}

\DeclareMathOperator{\polylog}{polylog}
\newcommand{\absolute}[1]{\lvert#1\rvert}

\newcommand{\floor}[1]{\lfloor#1\rfloor}

\newcommand{\ceil}[1]{\lceil#1\rceil}

\newcommand{\angles}[1]{\langle#1\rangle}

\newcommand{\rev}[1]{\textop{rev}(#1)}

\newcommand{\cent}{\textop{cen}}

\let\emptystring\varepsilon
\let\emptyseries\emptystring
\let\smallconst\epsilon

\theoremstyle{plain}

\newtheorem{algo}[theorem]{Algorithm}

\def\mathbfk{{\boldmath$k$\unboldmath}}
\title{Small Space Encoding and Recognition of \mathbfk{}-Palindromic Prefixes}

\author{Gabriel Bathie}
{DIENS, École normale supérieure de Paris, PSL Research University, France\\
LaBRI, Université de Bordeaux, France}
{gabriel.bathie@gmail.com}
{https://orcid.org/0000-0003-2400-4914}
{}

\author{Jonas Ellert}
{DIENS, École normale supérieure de Paris, PSL Research University, France}
{ellert.jonas@gmail.com}
{https://orcid.org/0000-0003-3305-6185}
{}

\author{Tatiana Starikovskaya}
{DIENS, École normale supérieure de Paris, PSL Research University, France}
{tat.starikovskaya@gmail.com}
{https://orcid.org/0000-0002-7193-9432}
{}

\authorrunning{G.\ Bathie, J.\ Ellert, and T.\ Starikovskaya}
\Copyright{Gabriel Bathie, Jonas Ellert, and Tatiana Starikovskaya}


 %TODO mandatory; please add comma-separated list of keystrings

\nolinenumbers %uncomment to disable line numbering
\hideLIPIcs %uncomment to compile preprint
\def\standalonetitle{0} %set to 1 for a separate title page

\begin{document}

\ifnum\standalonetitle=1 \thispagestyle{empty}\fi

\maketitle

\begin{abstract}
Palindromes are non-empty strings that read the same forward and backward. The problem of recognizing strings that can be represented as the concatenation of even-length palindromes, the concatenation of palindromes of length at least two, and the concatenation of exactly $k$ palindromes was introduced in the seminal paper of Knuth, Morris, and Pratt [SIAM J. Comput., 1977]. 

In this work, we study the problem of recognizing so-called \emph{$k$-palindromic} strings, which can be represented as the concatenation of exactly $k$ palindromes. It was shown that the problem is solvable in linear space and time [Rubinchik and Schur, MFCS'2020]. We aim to develop a sublinear-space solution, and show the following results: 

\begin{enumerate}
\item First, we show a structural characterization of the set of all $k$-palindromic prefixes of a string by representing it as a union of a small number of highly structured string sets, called \textit{affine prefix sets}. Representing the lengths of the $k$-palindromic prefixes in this way requires $\Oh(6^{k^2} \cdot \log^k n)$ space.
By constructing a lower bound, we show that the space complexity is optimal up to polylogarithmic factors for reasonably small values of $k$.

\item Secondly, we derive a read-only algorithm that, given a string $T$ of length $n$ and an integer~$k$, computes a compact representation of $i$-palindromic prefixes of $T$, for all $1 \le i \le k$. The algorithm uses $\Oh(n \cdot 6^{k^2} \cdot \log^k n)$ time and $\Oh(6^{k^2} \cdot \log^k n)$ space.

\item Finally, we also give a read-only algorithm for computing the palindromic length of $T$, which is the smallest $\ell$ such that $T$ is $\ell$-palindromic. Here, we achieve $\Oh(n \cdot 6^{\ell^2} \cdot \log^{\ceil{\ell/2}} n)$ time and $\Oh(6^{\ell^2} \cdot \log^{\ceil{\ell/2}} n)$ space. For some values of $\ell$, this is the first algorithm for palindromic length that uses $o(n)$ additional working space on top of the input.
\end{enumerate}
\end{abstract}

\ifnum\standalonetitle=1 \newpage\setcounter{page}{1}\fi

\section{Introduction}\label{sec:intro}

A \emph{palindrome} is a non-empty string that equals its reversed copy, i.e., a string that reads the same both forward and backward.
Throughout this work, we denote the language of palindromes by~$\pal$. We also define natural derivations of $\pal$: 
the language of even-length palindromes ${\pal_{\mathrm{ev}}}%
% = {\{P \in \pal : \absolute{P} \text{ is even}\}}%
$, the language of palindromes of length greater than one ${\pal_{>1}}%
%= {\{P \in \pal : \absolute{P} >1\}}%
$, and the language of concatenations of $k$ palindromes $\pal^k = \{P_1 P_2 \ldots P_k : P_i \in \pal, 1 \le i \le k\}$, for any $k \in \mathbb N^+$. The problem of recognising the languages $\pal_{\mathrm{ev}}^\ast$ (often referred to as ``palstar''), $\pal_{>1}^\ast$, and $\pal^k$, where one is given an input string and must decide whether it belongs to the language, is a classical question of formal language theory, initiated in the seminal paper of Knuth, Morris, and Pratt~\cite{DBLP:journals/siamcomp/KnuthMP77}.\footnote{${}^\ast$ is a Kleene star.}

Languages $\pal_{\mathrm{ev}}^\ast$, $\pal_{>1}^\ast$, and $\pal^k$ are context-free, and Valiant's parser from 1975 recognizes strings from these languages in $\Oh(n^\omega)$ time, where~$n$ is the length of the input string and $\omega$ is the matrix multiplication exponent. Only in 2018, Abboud, Backurs, and Vassilevska Williams showed that Valiant's parser is optimal if the current clique algorithms are optimal~\cite{DBLP:journals/siamcomp/AbboudBW18}, meaning that for general context-free languages, there is little hope of achieving a faster recognition algorithm.

The origins of the study of derivatives of the $\pal$ languages are in fact in line with the result of~\cite{DBLP:journals/siamcomp/AbboudBW18}: At one time, it was popularly believed that $\pal_{\mathrm{ev}}^\ast$ cannot be recognised in linear time, and it was considered as a candidate for a ``hard'' context-free language, as stated in the seminal paper of Knuth, Morris, and Pratt (see \cite[Section 6]{DBLP:journals/siamcomp/KnuthMP77}). However, Knuth, Morris, and Pratt~\cite{DBLP:journals/siamcomp/KnuthMP77} refuted this hypothesis by showing an $\Oh(n)$-time recognition algorithm for $\pal_{\mathrm{ev}}^\ast$. Manacher~\cite{10.1145/321892.321896} found another way to recognize $\pal_{\mathrm{ev}}^\ast$ in linear time, and Galil~\cite{10.1145/990502.990505} derived a real-time recognition algorithm (see also Slisenko~\cite{slisenko1981simplified}). Later, Galil and Seiferas~\cite{10.1145/322047.322056} showed a linear-time recognition algorithm for~$\pal_{>1}^\ast$. 

Recognition of $\pal^k$ appeared to be a much more intricate problem. 
Galil and Seiferas~\cite{10.1145/322047.322056} succeeded to design linear-time recognition algorithms for the cases $k = 1, 2, 3, 4$, but the general question remained open for almost 40 years. Only in 2015, Kosolobov, Rubinchik, and Shur~\cite{10.1007/978-3-662-46078-8_24} showed an $\Oh(nk)$-time recognition algorithm for $\pal^k$ for all $k \in \mathbb{N}^+$, which was finally improved to optimal $\Oh(n)$ time by Rubinchik and Shur in 2020~\cite{DBLP:conf/mfcs/RubinchikS20}. 
A related question is that of computing the \emph{palindromic length} of a string $T$, which is defined to be the smallest integer $k$ such that $T \in \pal^k$. The first $\Oh(n \log n)$-time algorithms for computing the palindromic length were presented in~\cite{DBLP:journals/jda/FiciGKK14,DBLP:conf/cpm/ISIBT14,DBLP:journals/ejc/RubinchikS18}. 
Finally, Borozdin, Kosolobov, Rubinchik, and Shur~\cite{borozdin2017linear} showed an optimal $\Oh(n)$-time algorithm for this problem. 

\subparagraph{Our contributions.}
In this work, we turn our attention to the \emph{space complexity} of recognising $\pal^k$ and computing the palindromic length. 
We start by presenting a characterization of prefixes of a given string that belong to~$\palk$. For $k=1$, we refer to these prefixes as \emph{prefix-palindromes}, and otherwise as \emph{$k$-palindromic prefixes}.
A crucial component of the linear time algorithm by Borozdin et al.~\cite{borozdin2017linear} is the following well-known property: 
the prefix-palindromes of a length-$n$ string can be expressed as $\Oh(\log n)$ arithmetic progressions.
If $x + a\cdot q$ with $a \in \{1,\dots ,u\}$ is such a progression, then there are strings $X[1\dd x]$ and $Q[1\dd q]$ such that $XQ^a$ 
%with $a \in \{1,\dots ,u\}$ 
is a prefix-palindrome%
, for every $a \in \{1,\dots ,u\}$
.
The arithmetic progression can be encoded in $\Oh(1)$ space, as it suffices to store $x$, $q$, and $u$.

In order to encode $k$-palindromic prefixes, we generalize arithmetic progressions to so-called \emph{affine prefix sets} of order $k$. Intuitively, such a set consists of prefixes of the form $XQ_1^{a_1}Q_2^{a_2}\dots Q_k^{a_k}$ with $\forall i \in [1, k] : a_i \in \{1,\dots ,u_i\}$. That is, rather than a single repeating substring $Q$, we allow multiple different substrings $Q_i$ of different lengths. An affine prefix set of order $k$ can then be encoded in $\Oh(k)$ space. By carefully analyzing the rich structure of periodic substrings induced by $k$-palindromic prefixes, we show that the $k$-palindromic prefixes can be expressed by a small number of affine prefix sets.

\begin{restatable}{theorem}{structuralthm}\label{th:structure}
    Let $0< \smallconst < 1$ be constant. Let $T[1 \dd n]$ be a string and let $k \in \mathbb N^+$. The set of prefixes of $T$ that belong to $\palk$ is the union of $\Oh(6^{k^2 / (2-\smallconst)} \cdot \log^k n)$ affine prefix sets, each of order at most $k$.
\end{restatable}

Surprisingly, this representation is within $\textnormal{polylog}(n)$-factors of the optimal encoding, at least for small values of $k$. We show the lower bound by explicitly constructing a large family of strings that can be uniquely identified by their palindromic prefixes:  

\begin{restatable}{theorem}{staticlowerbound}
\label{th:lower_bound}
%Any data structure for \staticprob on strings of length $n$ uses~$\Omega((\frac{\log n}{k})^k)$ bits of space.
Let $T[1 \dd n]$ be a string and let $k \in \mathbb N^+$. Encoding the lengths of the prefixes of $T$ that belong to $\pal^i$, for each $i \in [1,k]$, requires $\Omega(k^{-k} \cdot (\log_3 n)^k)$ bits of space.
\end{restatable}

As our final contribution, we derive a small-space read-only algorithm for constructing the affine prefix sets of \cref{th:structure}. Particularly, we show how to compute a small-space representation of the $i$-palindromic prefixes of $T$ for each $i\le k$. 
Recall that, in the read-only model of computation, one has constant-time random access to the input string. The space complexity of the algorithm is the space used beyond storing the input string.  

\begin{restatable}{theorem}{roalgopalk}\label{th:ro-algo-palk}
    Let $0 < \smallconst < 1$ be constant. Given a string $T[1\dd n]$ and $k \in \mathbb N^+$, there is a read-only algorithm that returns a compressed representation of all prefixes of $T$ that belong to $\pal^i$, for each $i \in [1, k]$, in $\Oh(n \cdot 6^{k^2/(2-\smallconst)} \cdot \log^k n)$ time and $\Oh(6^{k^2/(2-\smallconst)} \cdot \log^k n)$ space.
    %The compressed representation of the prefixes in $\pal^i$ consists of $\Oh(n \cdot 6^{i^2/(2-\smallconst)} \cdot \log^i n)$ canonical representations of affine prefix sets of $T$, each of order at most $i$.
\end{restatable}

As a corollary, we derive a parametrized read-only algorithm for computing the palindromic length.

\begin{restatable}{theorem}{palindromiclength}\label{th:algo-pal-length}
Given a string $T[1\dd n]$, there is a read-only algorithm that computes the palindromic length $k$ of $T$ in $\Oh(n \cdot 6^{k^2} \cdot \log^{\ceil{k/2}} n)$ time and $\Oh(6^{k^2} \cdot \log^{\ceil{k/2}} n)$ space.
\end{restatable}

In particular, for $k = \Oh(\log \log n)$, the algorithm uses $n\log^{\Oh(k)}n$ time and $\log^{\Oh(k)}n$ space, and for $k = o(\sqrt{\log n})$, it uses $n^{1+o(1)}$ time and sublinear $n^{o(1)}$ space. In the regime of small palindromic length, this is an improvement over all previously-known algorithms~\cite{borozdin2017linear,DBLP:conf/mfcs/RubinchikS20}, which require $\Omega(n)$ space. It remains an intriguing open question whether it is possible to achieve \emph{both} optimal linear time and sublinear space. 

Note however that the lower bound of \cref{th:lower_bound} does not imply a lower bound for our read-only algorithm as it has access to the input, and even more so for an algorithm that computes only the palindromic length of the input. On the other hand, proving an $\Omega(\log^{f(k)} n)$ space lower bound for a read-only algorithm might be well beyond the current techniques: the only lower bound technique for read-only  string processing algorithms the authors are aware of is based on deterministic branching programs~\cite{10.1145/800141.804677}, and it shows that any read-only algorithm for computing the longest common substring of two strings that works in sublinear space must use slightly superlinear time; formally, an algorithm that uses $\Oh(\tau)$ space requires $\Omega(n \sqrt{\log (n /\tau \log n)/\log \log (n /\tau \log n)})$ time~\cite{DBLP:conf/esa/KociumakaSV14}. 

\medskip

\subparagraph{Related work.} Berebrink, Erg\"{u}n, Mallmann-Trenn, and Azer~\cite{DBLP:conf/stacs/BerenbrinkEMA14} initiated a study of the space complexity of computing the longest palindromic substring of a string in the streaming model, where the input arrives symbol by symbol and the space complexity is defined as the total space used, including any information an algorithm stores about the input. They developed the first 
trade-offs between the bound on the error and the space complexity for approximating
the length of the longest palindrome with either additive or multiplicative error, which were sequentially tightened by Gawrychowski, Merkurev, Shur, and Uzna\'{n}ski~\cite{DBLP:journals/algorithmica/GawrychowskiMSU19}. 

Amir and Porat~\cite{DBLP:conf/cpm/AmirP14} gave the first streaming algorithm for computing all approximate prefix-palindromes of a string. Namely, given an integer parameter $k$, their algorithm computes all prefixes within Hamming distance $k$ from $\pal$. Bathie, Kociumaka, and Starikovskaya~\cite{DBLP:conf/isaac/BathieKS23} improved their result for the Hamming distance and expanded it to the edit distance and the read-only setting.

The problem of recognising formal languages in small space has been also studied for regular languages, see~\cite{GHKLM18,GHL16,GanardiHL18,ganardi_et_al:LIPIcs:2019:11502,DBLP:conf/soda/DudekGGS22}, the Dyck language (the language of well-parenthesized expressions), see~\cite{JN14,KrebsLS11,MagniezMN14}, for visibly pushdown languages (a language class strictly in-between the regular and context-free languages with good closure and decidability properties~\cite{10.1145/1007352.1007390}), see~\cite{francois_et_al:LIPIcs:2016:6355,G19,DBLP:conf/icalp/BathieS21}, general context-free languages~\cite{DBLP:conf/mfcs/GanardiJL18}, and for $\mathsf{DLIN}$ and $\mathsf{LL}(k)$, see~\cite{BabuLRV13}.

\subparagraph*{Roadmap.} The remainder of the paper is structured as follows. In \cref{sec:prelim}, we introduce notation, basic definitions, and auxiliary lemmas. We then show the lower bound for encoding palindromic prefixes in \cref{sec:lower-bound}. The space efficient encoding is presented in two steps. First, in \cref{sec:affineprefixsets}, we describe affine prefix sets, their fundamental properties, and how they are related to the structure of periodic substrings. Then, in \cref{sec:appendapalindrome}, we show how to encode the $k$-palindromic prefixes using affine prefix sets of order $k$, inductively assuming that the $(k-1)$-palindromic prefixes are already given as a union of affine prefix sets of order $k - 1$. 
Finally, the algorithms from \cref{th:ro-algo-palk,th:algo-pal-length} are described in \cref{sec:algo}.

%%%%%%%%%%%%%%%%%%%%%%%%
%%%%%%%%%%%%%%%%%%%%%%%%
\section{Preliminaries}
\label{sec:prelim}

\subparagraph*{Series, strings, and substrings.} For $i, j \in \mathbb Z$, we write $[i, j] = [i, j + 1) = {(i - 1, j]} = {(i - 1, j + 1)}$ to denote
%the integer range 
$\{ h \in \mathbb Z \mid i \le h \le j \}$.
A series $a_1, b_1, c_1, a_2, b_2, c_2, \dots, a_t, b_t, c_t$ is denoted by $(a_i, b_i, c_i)_{i=1}^{t}$. The empty series is denoted by $\emptyseries$. We use the dot-product to denote the concatenation of two series, e.g., $(a_i, b_i, c_i)_{i=1}^{t} = (a_i, b_i, c_i)_{i=1}^{t-3} \cdot (a_i, b_i, c_i)_{i=t-2}^{t}$. We may omit the subscript and superscript for series of length one, e.g., $(a_1, b_1, c_1) = (a_i, b_i, c_i)_{i=1}^{1}$.

A \emph{string} $T$ of length $\absolute{T} = n$ is a sequence of $n$ \emph{symbols} from a set $\Sigma$, which we call the \emph{alphabet}. The input string is also called \emph{the text}.
We denote the set of all length-$n$ strings by $\Sigma^n$, and we set $\Sigma^{\le n} = \bigcup_{m=0}^n \Sigma^m$ as well as $\Sigma^* = \bigcup_{n= 0}^\infty \Sigma^n$. 
The empty string is denoted by $\emptystring$. For $i, j \in [1, n]$, the $i$-th symbol in $T$ is denoted by $T[i]$. The \emph{substring} $T[i\dd j] = T[i\dd j + 1) = T(i - 1\dd j] = T(i - 1\dd j + 1)$ is the empty string $\emptystring$ if $j < i$, and the string $T[i]T[i + 1]\dots T[j]$ otherwise.
We may call a substring $T[i\dd j]$ a \emph{fragment} of $T$ to emphasize that we mean the specific occurrence of $T[i\dd j]$ that starts at a position $i$. 
For example, in the string $T = \texttt{abcabc}$, the substrings $T[1\dd 3]$ and $T[4\dd 6]$ are identical, but $T[1\dd 3]$ and $T[4\dd 6]$ are distinct fragments.
A string $S$ is a \emph{prefix} of $T$ if there is $i \in [1, n]$ such that $S = T[1\dd i]$, in which case we may simply write $T[\dd i]$.
Similarly, $S$ is a \emph{suffix} of $T$ if there is $i \in [1, n]$ such that $S = T[i\dd n]$, in which case we may simply write $T[i\dd ]$.
We extend this notion to the empty suffix $T[n+1\dd n] = T[n + 1\dd ]$ and the empty prefix $T[1\dd 0] = T[\dd 0]$.
A substring (hence also a suffix or prefix) of $T$ is \emph{proper} if it is shorter than $T$.
When introducing a string $S$, we may simply say that $S[1\dd m]$ is a string rather than saying that $S$ is a string of length $m$.
The concatenation of two strings $S[1\dd m]$ and $T[1\dd n]$ is the string $S[1]S[2]\dots S[m]T[1]T[2]\dots T[n]$, denoted by either $S \cdot T$ or simply $ST$. For non-negative integer $a$, we write $T^a$ to denote the length-$(an)$ string obtained by concatenating $a$ copies of $T$.
We extend this idea to non-negative rational exponents $\alpha \in \mathbb Q$, for which we write $T^\alpha$ to denote $T^{\floor{\alpha}} \cdot T[1\dd (\alpha n \bmod n)]$. We only use this notation if $\alpha n \in \mathbb N$. For example, $(\texttt{abcdef})^{4/3} = \texttt{abcdefab}$.

\subparagraph*{Palindromes and periodicities.} For a string $T[1\dd n]$, we write $\rev{T}$ to denote its reverse, i.e., $\rev{T} = T[n]T[n-1]\cdots T[1]$. We then say that $T$ is a palindrome if and only if $T$ is non-empty and $T = \rev{T}$.
The set of all palindromes is denoted by $\pal$. For a positive integer $k$, the set $\palk$ contains all the strings that can be written as the concatenation of exactly $k$ palindromes. We refer to such strings as \emph{$k$-palindromic}. If $k=1$, and a string is a one-palindromic prefix of another string, we also refer to it as \emph{prefix-palindrome}. 

We define the \emph{forward cyclic rotation} $\rot(T) = T[2\dd n]T[1]$.
More generally, a cyclic rotation $\rot^s(T)$ with \emph{shift} $s \in \mathbb{Z}$
is obtained by iterating $\rot$ (if $s$ is positive) or the inverse operation $\rot^{-1}$ (if $s$ is negative) exactly $\absolute{s}$ times.
A non-empty string $T[1\dd n]$ is \emph{primitive} if it is distinct
from its non-trivial rotations, i.e., if $T = \rot^s(T)$ holds only when $n$ divides~$s$.
Equivalently, a string $T$ is primitive if it cannot be written as $T = S^a$, for any string $S$ and integer $a \ge 2$.

A string $T[1\dd n]$ has \emph{period} $p \in \mathbb N^+$ if $\forall i \in [1, n - p] : T[i] = T[i + p]$, or equivalently if $T[1\dd n - p] = T(p\dd n]$. The string $T[1\dd n - p] = T(p\dd n]$ is a \emph{border} of $T$. If $T$ has period $p \le n/2$, then we say that $T$ is $p$-periodic.
If $T$ has period $p \le n$, then it can be written as $T = P^{\floor{n / p}}P[1\dd n \bmod p]$, where $P = T[1\dd p]$. We may alternatively use a rational exponent and write $T = P^{{n / p}}$.
A fundamental tool for analyzing periodicities is the celebrated periodicity lemma by Fine and Wilf, which is stated below. 
Additionally, we provide some simple auxiliary lemmas regarding periodic string and palindromes.

\begin{lemma}[Periodicity Lemma \cite{finewilf}]\label{lem:finewilf}
	If $p$ and $q$ are distinct periods of a string of length at least $p + q - \gcd(p,q)$, then $\gcd(p, q)$ is a period of the string.
\end{lemma}

\begin{lemma}%
        \label{lem:primitive_squares}%
	For a primitive string $Q$, the minimal period of $Q^2$ is $\absolute{Q}$.
\end{lemma}
\begin{proof}
	Let $q = \absolute{Q}$. If $Q^2$ has period $q' < q$, then \cref{lem:finewilf} implies that $p = \gcd(q', q) < q$ is a period of $Q^2$, hence also of $Q$. Then, however, it holds $Q = Q[1\dd p]^{q / p}$ with $1 < (q / p) \in \mathbb N$, which contradicts the fact that $Q$ is primitive. 
\end{proof}

\begin{lemma}\label{lem:aux:structure:palperiodoverlap}
    If a palindrome $P$ has a $q$-periodic prefix of length $m \ge 3q/2$, then either $P$ is $q$-periodic, or $\absolute{P} > 2m - q$.
\end{lemma}

\begin{proof}
    Since $P$ is a palindrome, it does not only have a $q$-periodic prefix of length $m$, but also a $q$-periodic suffix of length $m$. If $\absolute{P} \le 2m - q$, then the periodic prefix overlaps the periodic prefix by at least $q$ symbols, and hence it is easy to see that the entire $P$ is $q$-periodic.
\end{proof}

\def\qsuff{Q_{\textsf{suff}}}
\def\qpref{Q_{\textsf{pref}}}

\def\qsufflem{{S}}
\def\qpreflem{{P}}
\def\qplainlem{{Q}}

\begin{lemma}\label{lem:aux:structure:palindrome_in_run}
Let $\qplainlem$ be a string with suffix $\qsufflem$ and prefix $\qpreflem$. Let $x$ be a positive integer, then ${\qsufflem}\qplainlem^x{\qpreflem}$ is a palindrome if and only if $\rot^{\absolute{\qpreflem} - \absolute{\qsufflem}}(\qplainlem) = \rev{\qplainlem}$.
\end{lemma}

\begin{proof}
\let\qsuff\qsufflem
\let\qpref\qpreflem
\let\plainQ\qplainlem
\def\hatQ{{\overline\plainQ}}
\def\hatqsuff{{\overline {\qsuff}}}
\def\hatqpref{{\overline {\qpref}}}
    Let $q = \absolute{\plainQ}$, $s = \absolute{{\qsuff}}$, $p = \absolute{{\qpref}}$, $\hatQ = \rev{\plainQ}$, $\hatqsuff = \rev{{\qsuff}}$, and $\hatqpref = \rev{{\qpref}}$.
    We assume $s \le p$ and show that ${\qsuff}\plainQ^x{\qpref}$ is a palindrome if and only if $\rot^{p-s}(\plainQ) = \rev{\plainQ}$. The individual steps are explained below.
    \begin{alignat*}{4}
        \hatqpref \cdot \hatQ^x \cdot \hatqsuff \enskip  = \enskip {\qsuff}\cdot \plainQ^x\cdot {\qpref} %
        \enskip & {\iff}_{(i)} && \hatqpref \cdot \hatQ^x && = && {\qsuff}\cdot \plainQ^x\cdot {\qpref}[1\dd p - s]\\ 
        &{\iff}_{(ii)} && \hatqpref \cdot \hatQ && = && {\qsuff}\cdot \plainQ[1\dd p-s]\cdot \rot^{p-s}(\plainQ)\\
        &{\iff}_{(iii)}\enskip && \hatQ &\enskip& = &\enskip& \rot^{p - s}(\plainQ).
    \end{alignat*}
    For $(i)$, the  $(\Rightarrow)$-direction is trivial, as we merely trim a suffix of length $s$ on both sides. For the $(\Leftarrow)$-direction, $\hatqpref \cdot \hatQ^x = {\qsuff}\cdot \plainQ^x\cdot {\qpref}[1\dd p - s]$ with $s \le p$ implies that ${\qsuff}$ is a prefix of $\hatqpref$, which means that $\hatqsuff$ is a suffix of ${\qpref}$. Hence $\hatqsuff = {\qpref}[p-s+1\dd p]$, and the $(\Leftarrow)$-direction of $(i)$ follows.
    Step $(ii)$ is straight-forward.
    Finally, the $(\Rightarrow)$-direction of $(iii)$ is again trivial, as we merely trim a prefix of length $p$
    %and a suffix of length $qx - q$ 
    on both sides. For the $(\Leftarrow)$-direction, note that 
    $\hatQ = \rot^{p - s}(\plainQ)$ 
    can be written as 
    $\hatQ[1\dd q-p] \cdot \hatqpref = \plainQ[p-s+1\dd q] \cdot \plainQ[1\dd p-s]$.
    By trimming a prefix of length $q - p$ on both sides, we obtain
    ${\hatqpref = \plainQ[q-s+1\dd q] \cdot \plainQ[1\dd p-s] = {\qsuff}\plainQ[1\dd p-s]}$.
    Hence $\hatQ = \rot^{p - s}(\plainQ)$ implies ${\hatqpref \cdot \hatQ = {\qsuff}\plainQ[1\dd p-s]\cdot \rot^{p - s}(\plainQ)}$, and the $(\Leftarrow)$-direction of $(iii)$ follows.

	If $s > p$, then we simply invoke the lemma with $\plainQ' = \hatQ$, which has suffix $\qsuff' = \hatqpref$ of length $s' = p$ and prefix $\qpref' = \hatqsuff$ of length $p' = s$ with $s' < p'$. We have already shown the correctness of the lemma for this case. Thus, $\qsuff'(\plainQ')^x\qpref' = \rev{\qsuff\plainQ^x\qpref}$ is a palindrome if and only if $\rot^{p' - s'}(\plainQ') = \rev{\plainQ'}$, which is equivalent to $\rot^{-(s - p)}(\plainQ) = \plainQ$.
\end{proof}

\begin{corollary}%
\label{lem:aux:structure:palindrome_in_run:extended}%
Let $y$ be a positive integer, let $\qplainlem$ be a string, and let $\qsufflem$ and $\qpreflem$ be respectively a suffix and a prefix of $\qplainlem^y$. Let $x$ be a positive integer, then ${\qsufflem}\qplainlem^x{\qpreflem}$ is a palindrome if and only if $\rot^{\absolute{\qpreflem} - \absolute{\qsufflem}}(\qplainlem) = \rev{\qplainlem}$.
\end{corollary}

\begin{proof}
    Let $q = \absolute{\qplainlem}$, $s = ({\absolute{\qsufflem} \bmod q})$, $p = ({\absolute{\qpreflem} \bmod q})$, $x' = \floor{\absolute{\qsufflem} / q}$ and $x'' = \floor{\absolute{\qpreflem} / q}$. 
    By \cref{lem:aux:structure:palindrome_in_run}, $\qsufflem\qplainlem^x\qpreflem = \qplainlem[q - s + 1\dd q]\qplainlem^{x' + y + x''}\qplainlem[1\dd p]$ is a palindrome if and only $\rot^{p - s}(\qplainlem) = \rev{\qplainlem}$.
    Note that $\absolute{\qpreflem} - \absolute{\qsufflem} = q \cdot (x'' - x') + p - s$, and thus $\rot^{\absolute{\qpreflem} - \absolute{\qsufflem}}(\qplainlem) = \rot^{q \cdot (x'' - x') + p - s}(\qplainlem) = \rot^{q \cdot (x'' - x')}(\rot^{p - s}(\qplainlem)) = \rot^{p - s}(\qplainlem)$, which concludes the proof.
\end{proof}

\subparagraph*{Model of computation.}
We assume the standard word RAM model of computation~\cite{fredmanwillard}, using words of size $\Theta(\log n)$ when processing an input string of length $n$. The presented algorithms are deterministic and read-only, i.e., they cannot write to the memory occupied by the input. Space complexities are stated in number of words, ignoring the space occupied by the input.

%%%%%%%%%%%%%%%%%%%%%%%%
%%%%%%%%%%%%%%%%%%%%%%%%
\section{Lower Bound}\label{sec:lower-bound}

We now show that the size of our representation of $k$-palindromic prefixes is close to optimal when $k$ is small. For constant $k$, we are within an $\Oh(\log n)$ factor of the optimal space (where the lower bound is $\Omega(\log^k n)$ \emph{bits}, while our representation requires $\Oh(\log^k n)$ \emph{words}). For $k = \Oh(\sqrt{\log \log n})$, we are within an $\Oh(\polylog(n))$ factor of the lower bound.
 
\staticlowerbound*

For some string $X$ and positive integer $s$, let $\palprefs(X) \subseteq [1, \absolute{X}] \times [1, s]$ be defined such that $(i, r) \in \palprefs(X)$ if and only if $r \leq s$ is the palindromic length of $X[..i]$. We will provide a lower bound for the space needed to encode $\palprefs(X)$.
For this purpose, we construct for any integers $t, s \ge 1$
a family of strings $F(t, s)$ with the following properties:

\begin{enumerate}
    \item for every $X\in F(t, s), \absolute{X} = 3^{t+s}$, and
    \item we have $\absolute{F(t, s)} \ge 2^{b(t, s)}$,
    where $b(t, s) \ge \binom{t+s}{s}$, and
    \item the function $f_{s}(X) = \palprefs(X)$ is injective on $F(t,s)$, and its inverse is computable.
\end{enumerate}

Due to the second property, we can bijectively map all the possible bitstrings of length $b(t, s)$ to a subset of $F(t, s)$, and hence encoding an element of $F(t, s)$ requires at least $b(t, s)$ bits of space in the worst case.
Assume that, for some string $X \in F(t, s)$, we are given the sets $\pal^i(X)$ for every $i \leq s$.
From these sets, we can easily construct $\palprefs(X)$. By the third property, we can then compute $X = f_{s}^{- 1}(\palprefs(X))$. Hence the sets uniquely encode $X$, which implies that they cannot be stored in fewer than $b(t, s)$ bits of space.
If we let $s = k$ and use $n = 3^{t+k}$ for any $t\ge 1$, 
this bound is at least $\binom{\log_3 n}{k} \geq ((\log_{3} n) / k)^k$ bits.

%Note that $\pal^i(X)$ can be obtained from $\pal^{k}(\texttt\textdollar_1\texttt\textdollar_2\dots\texttt\textdollar_{k-i} \cdot X \cdot \texttt\textdollar_{k - i + 1}\dots\texttt\textdollar_k)$, hence the prefixes in $\pal^k$ of a string of length $n$ cannot be 
%(The choice of $s$ and $n$ implies $k \leq (\log_3 n) - 1$. If $k > (\log_3 n) - 1$ and $n \geq 9$, then $k > \log_{10} n$, and the claimed lower bound is constant.)

\paragraph*{Recursive construction of \boldmath$F(t,s)$.\unboldmath} For $t\ge 1$ and $s\ge 1$, we define%
\begin{flalign*}
    F(t, 1) &= \{a^{i}b^{3^{t+1}-2i}a^{i}: i = 1,\ldots, 2^{t+1}\} \text{ where $a,b$ are distinct letters in }
     \Sigma\text{, and}\\ 
    F(1, s)&= \{UVU : U, V \in F(1, s-1)\} \text{ when } s\ge 2\text{, and}\\ 
    F(t, s) &= \{UVU : U\in F(t-1, s) \land V\in F(t, s-1)\} \text{ when } t, s \ge 2.
\end{flalign*}
Note that $F(t, 1)$ is well-defined, as $t\ge 1$, hence $2^{t+2} < 3^{t+1}$.
In the above definition of $F(1, s)$ and $F(t, s)$.

\paragraph*{Properties of \boldmath$F(t,s)$.\unboldmath} We now prove that $F(s,t)$ has the desired properties.
The following claim can easily be proven by induction.
\begin{claim}
    For every $X\in F(s,t), \absolute{X} = 3^{t+s}$ and $X\in\pal$. 
\end{claim}

\begin{claim}
    Let $b(t, s) = \log_2 \absolute{F(t,s)}$. We have $b(t,s) \ge \binom{t+s}{s}$.
\end{claim}
\begin{claimproof}
    We proceed by induction.
    For $t \geq 1$, it holds $\absolute{F(t, 1)} = 2^{t+1}$, i.e., $b(t,1) = t+1 = {\binom{t+1}{1}}$. 
    In particular, $\absolute{F(1, 1)} = 2^{2}$ and $b(1,1) = 2 = {\binom{2}{1}}$.
    For $s \geq 2$, it clearly holds $\absolute{F(1, s)} = \absolute{F(1, s-1)}^2$, and thus $b(1, s) = 2 \cdot b(1, s - 1) \geq 2 \cdot \binom{s}{s - 1} = 2s \geq s + 1 = \binom{s + 1}{s}$ by immediate induction. It remains the case where $t \geq 2$ and $s \geq 2$, in which we have $\absolute{F(t, s)} = \absolute{F(t-1, s)} \cdot \absolute{F(t, s - 1)}$ and thus%
    %
    %If $s = 1$, we have $\absolute{F(t, 1)} = 2^{t+1}$, i.e., $b(t,s) = t+1 = {\binom{t+s}{s}}$.
    %In particular, for $s = 1$, $\absolute{F(1, s)} = 2^{2} = 2^{s+1}$.
    %Therefore, if $t = 1$ and $s \ge 2$, we have $\absolute{F(1, s)} = 2\absolute{F(1, s-1)} \ge 2^{s+1}$,
    %by immediate induction, and the claim holds as $2^{s+1} = 2^{\binom{t+s}{s}}$ when $t = 1$. Otherwise, if $t,s \ge 2$, we have:
    \begin{flalign*}
        b(t, s) = b(t-1, s) + b(t, s-1) \ge \binom{t+s-1}{s} + \binom{t+s-1}{s-1}%
        %\text{ (by induction)} 
        = \binom{t + s}{s}%
        ,%
    \end{flalign*}
    where the second step is by induction. Hence the claim holds.
    %This shows that the claim inductively holds.
\end{claimproof}

Finally, we show that every $X \in F(t,s)$ can be uniquely identified from $\palprefs(X)$.
\begin{claim}
    The function $f_{s}(X) = \palprefs(X)$ is injective on $F(t,s)$,
    and its inverse is computable.
\end{claim}
\begin{claimproof}
    In accordance with the recursive definition of $F(t,s)$, we proceed by induction.
    We start with the base case $s = 1$. The elements of $F(t,1)$ are of the form $X_i = a^{i}b^{3^{t+1}-2i}a^{i}$
    for some $i \le 2^{t+1}$.
    It follows that $(j,1)\in \palpref^1(X_i)$ for all $j\le i$, and $(i+1, 1)\notin\palpref^1(X_i)$.
    Hence $f_{1}$ is injective on $F(t, 1)$, and we can compute the inverse $X_i = f_{1}^{-1}(\palpref^1(X_i))$ due to $i = \min \{j \in \mathbb N^+ \mid (j + 1, 1) \notin \palpref^1(X_i) \}$.

    Next, assume that $t = 1$ and $s \ge 2$, and we inductively assume that $f_{s - 1}$ is injective on $F(1, s - 1)$. Let $X,Y\in F(1,s)$, and assume that $\palprefs(X) = \palprefs(Y)$.
    Note that there exist strings $U, U', V, V' \in F(1, s-1)$ such that
    $X = UVU$ and $Y = U'V'U'$.
    We show $U = U'$ and $V = V'$, which implies $X = Y$.
    Note that, as $U$ is a prefix of $X$, $\palpref^{s-1}(U)$ can be computed from $\palprefs(X)$ using
    \[\palpref^{s-1}(U)= \palprefs(X) \cap ([1\dd |U|] \times [1\dd s-1]).\]
    The same equation holds for $U'$ and $Y$, hence $\palpref^{s-1}(U) = \palpref^{s-1}(U')$.
    We inductively assume that $f_{s-1}(X) = \palpref^{s-1}(X)$ is injective on $F(1, s-1)$,
    and it follows $U = U'$.

    We now turn to showing that $V = V'$.
    Consider an index $i \in [|U|+1\dd |UV|]$, and let $s'$ be the palindromic length of $X[..i] = U \cdot V[..i - \absolute{U}]$. Recall that $U$ is a palindrome, and that $U$ and $V$ are encoded over distinct alphabets. Consequently, $s' - 1$ is the palindromic length of $V[..i - \absolute{U}]$.
    Therefore, $\palpref^{s-1}(V)$ can also be derived from $\palprefs(X)$:
    \[\palpref^{s-1}(V)= \{(i - \absolute{U}, s' - 1) : (i, s') \in \palprefs(X)\} \cap ([1\dd \absolute{V}] \times [1\dd s - 1]).\]
    As the same relation holds for $V'$ and $Y$, it follows that
    $\palpref^{s-1}(V) = \palpref^{s-1}(V')$, and by induction hypothesis, this implies that $V = V'$,
    and therefore $X = Y$.
    The inverse function $f_{s}^{-1}$ can be efficiently computed using two recursive calls to $f_{s-1}^{-1}$.

    The proof for the case when $t, s \ge 2$ is highly similar, but the induction is over the sum $t + s$. We inductively assume that $f_{s - 1}$ is injective on $F(t, s - 1)$, and that $f_{s}$ is injective on $F(t - 1, s)$. Instead of $\palpref^{s-1}(U)$, we consider $\palprefs(U)$. The remainder of the proof is the same as for $t = 1$ and $s \geq 2$.
\end{claimproof}
This concludes the proof of \cref{th:lower_bound}.

%%%%%%%%%%%%%%%%%%%%%%%%%%
\section{Combinatorial Properties of Affine Prefix Sets}
\label{sec:affineprefixsets}

In this section, we study the combinatorial structure of $k$-palindromic prefixes of $T$.
We start with the definition of \textit{affine sets}, which we will use as a scaffolding for our analysis.

\begin{restatable}[Affine sets]{definition}{affinesetdef}\label{def:affine-sets}
    A set of strings $\Aa$ is \emph{affine} if there exist $t \in \mathbb N_0$, a string $X$, primitive strings $Q_1,\ldots, Q_t$, and positive integers $\ell_1, \dots, \ell_t$ and $u_1,\dots, u_t$ such that $$\forall i \in [1, t] : \ell_i \le u_i\textnormal{\qquad and\qquad}\Aa = \{XQ_1^{a_1}\ldots Q_t ^{a_t} \mid \forall r \in [1, t] : a_r \in [\ell_r, u_r]\}.$$
    The tuple $\angles{X,\left(Q_i, \ell_i, u_i\right)_{i=1}^t}$ is a \emph{representation} of $\Aa$, and $t$ is the \emph{order} of the representation. The order of $\Aa$ is the minimal order achieved by any of its representations. We call $\{Q_i\}$ \emph{the components} of a representation, and $\ell_i$ (resp., $u_i$) the exponent lower (resp., upper) bounds.
\end{restatable}

A representation \emph{generates} (the strings of) the corresponding affine string set. 
If $\angles{X,\left(Q_i, \ell_i, u_i\right)_{i=1}^t}$ generates $\Aa$ and $\angles{X',\left(Q_i', \ell_i', u_i'\right)_{i=1}^{t'}}$ generates $\Bb$, 
then their concatenation is defined as $\angles{X,\left(Q_i, \ell_i, u_i\right)_{i=1}^t \cdot (Y, a, a) \cdot \left(Q_i', \ell_i', u_i'\right)_{i=1}^{t'}}$, 
where $Y$ is a primitive string and $a$ is a positive integer such that $Y^a = X'$ (i.e., $Y$ is the primitive root of $X'$). 
The concatenation generates $\Aa \cdot \Bb = \{A\cdot B: A\in\Aa \land B\in\Bb\}$. (If $X' = \emptystring$, then the concatenation is $\angles{X,\left(Q_i, \ell_i, u_i\right)_{i=1}^t \cdot \left(Q_i', \ell_i', u_i'\right)_{i=1}^{t'}}$.)

\FloatBarrier

\begin{figure}
	\centering\def\vspacer{\vspace{1.5\baselineskip}}

	\subcaptionbox{String $T$ with prefix-palindromes $\{\texttt{a}, \texttt{aba}, \texttt{ababa}, \texttt{ababaccababa}, \texttt{ababaccababaccababa}\}$. The prefix-palindromes can be expressed using four affine prefix sets with representations $\angles{\texttt{a}, \emptyseries}$, $\angles{\texttt{aba}, \emptyseries}$, $\angles{\texttt{a}, (\texttt{ba},1,2)}$, and $\angles{\texttt{ababa}, (\texttt{ccababa}, 1, 2)}$. The same number of affine prefix sets suffices to express the prefixes of $T$ in $\pal^2$, as shown below $T$. There are $14$ such prefixes, and some of them can be split into two palindromes in multiple ways.\label{fig:pal2}}[\textwidth]{
	\def\collla{white}
\def\colllb{black!10!white}
\def\colllc{black!20!white}
\def\colll{white}
\begin{tikzpicture}[x=.89em, y=1em]

\foreach[
	count=\iminus from 0, 
	count=\i from 1,
	evaluate=\iminus as \lll using \iminus-0.5,
	evaluate=\iminus as \rrr using \iminus+0.5,
	] \x in
{a,b,a,b,a,c,c,a,b,a,b,a,c,c,a,b,a,b,a,c,c,a,b,a,b,} {

	\edef\tmpcol{\csname colll\x\endcsname}
	\fill[\tmpcol] (\lll, 4) rectangle (\rrr, -13.25);

	\node[inner sep=0] (txt\i) at (\iminus, 0) {\texttt{\x}\vphantom{\texttt{h}}};

}

\node[fit=(txt1)(txt25), draw] {};

\node[left=0 of txt1] {$T =\ $\vphantom{\texttt{h}}};

\xdef\hidx{1}
\foreach[count=\j from 0] \group/\splitsize in {%
		{%
		1/10,%
		1/17,%
		1/24%
		}/0,%
		{%
		3/20,%
		3/13,%
		3/6%
		}/0.5,%
		{%
		3/1,%
		1/3,%
		1/1%
		}/2,%
		{%
		19/2,%
		12/9,%
		5/16,%
		5/9,%
		12/2%
		}/3.5,%
		{%
		19/1,%
		12/1%
		}/4,%
		{%
		5/2,%
		5/1%
		}/5.5,%
		{%
		1/0,%
		3/0,%
		5/0,%
		12/0,%
		19/0%
		}/-28.5%
}
{

\foreach[
	count=\i from \hidx,
	count=\iplus from \hidx+1,
	evaluate=\i as \voff using (0.5*\i+0.5*\splitsize+0.5)),
	evaluate=\first as \secstart using int(\first + 1),
	evaluate=\first as \secend using int(\first+\second),
	evaluate=\secend as \secendplus using int(\secend+1)] 
	\first/\second in \group {
	\node[below=\voff of txt1] (vtmp) {};
	\path (txt\first.center) to node[midway] (split) {} (txt\secstart.center);
	\path (txt\secend.center) to node[midway] (end) {} (txt\secendplus.center);
	\draw[|-|] (txt1.west |- vtmp) to (split |- vtmp);
	\ifnum\second>0 \draw[|-|] (split |- vtmp) to (end |- vtmp);\fi
	\node[inner sep=.2em] (l\hidx) at (txt1.west |- vtmp) {};
	\xdef\hidx{\iplus}
}
}

\draw[decorate, decoration = {brace, amplitude=.5em, raise=.5em}] (l6.south) to 
node[midway, left=2em] {$\angles{\emptystring, (\texttt{ababacc}, 1, 3) \cdot (\texttt{ab}, 1, 2)}$} (l1.north);
\draw[decorate, decoration = {brace, amplitude=.5em, raise=.5em}] (l9.south) to 
node[midway, left=2em] {$\angles{\emptystring, (\texttt{ab}, 1, 2)}$} (l7.north);
\draw[decorate, decoration = {brace, amplitude=.5em, raise=.5em}] (l16.south) to 
node[midway, left=2em] {$\angles{\texttt{ababa}, (\texttt{ccababa}, 1, 2) \cdot (\texttt{c}, 1, 2)}$} (l10.north);
\draw[decorate, decoration = {brace, amplitude=.5em, raise=.5em}] (l18.south) to 
node[midway, left=2em] {$\angles{\texttt{ababa}, (\texttt{c}, 1, 2)}$} (l17.north);
\draw[decorate, decoration = {brace, amplitude=.5em, raise=.5em}] (l23.south) to 
node[midway, left=2em] {prefix-palindromes of $T$} (l19.north);

\end{tikzpicture}}

\vspacer

	\subcaptionbox{An affine prefix set $\Aa$ of a string $T$ with representation \irreducibleset{} (drawn above $T$).
     This representation is irreducible.
     The set $\Aa$ contains all the prefixes of $T$ that end at positions drawn in dotted lines. In this example, the set $\Aa$ has the alternative representation \reducibleset. This representation is reducible because $Q'$ has the same exponent upper and lower bound, and because $Q_2$ has an exponent lower bound larger than $1$.\label{fig:aps}}[\textwidth]{
	\def\makeirreducibleset{\affinesetup[(irreduciblestart)]{%
		14/0/0/black/$X$,
		32/1/2/black/$Q_1$,
		5/1/3/black/$\ \ Q_2$,
		2/1/2/black/{}%
		}
		\xdef\irreducibleset{$\angles{X, %
     		(Q_1, \affinelower{1}, \affineupper{1}) \cdot %
     		(Q_2, \affinelower{2}, \affineupper{2}) \cdot %
     		(Q_3, \affinelower{3}, \affineupper{3})}$}
	}

	\def\makereducibleset{\affinesetup{%
		5/0/0/gray/$X'$,
		32/1/2/gray/$Q_1'$,
		4/1/1/gray/$\enskip\ Q'$,
		5/2/4/gray/$\ Q_2$,
		2/1/2/gray/{}%
		}
		\xdef\reducibleset{$\angles{X', %
     		(Q_1', \affinelower{1}, \affineupper{1}) \cdot %
     		(Q', \affinelower{2}, \affineupper{2}) \cdot %
     		(Q_2, \affinelower{3}, \affineupper{3}) \cdot %
     		(Q_3, \affinelower{4}, \affineupper{4})}$}
	}

    \begin{tikzpicture}[x=.37em, y=1em]	
    
    \makereducibleset
	\affinemakebelow
	\affinedrawing[{\draw[densely dotted] (affinemax) to ++(0, .5em);}]
	\draw (affinemax4-center) node[below, gray] {$\enskip Q_3$};
	
	\draw (affinemin) ++(0, .5em) node (bottomleft) {};
	\draw (affinemax |- bottomleft) ++(1em, 1em) node (topright) {};
	\node[draw, fit=(bottomleft.center)(topright.center), inner sep=0] (text) {};
	\node[left=0 of text] {$T =$};
	\draw (affinemin |- topright) ++(0, .5em) node (irreduciblestart) {}; 	
	
	{\tikzset{every path/.style={line width=1pt}}\makeirreducibleset
	\affinedrawing[{\draw[densely dotted] (affinemax) to ++(0, -1.5em);}]
	\draw (affinemax3-center) node[above, black] {$\enskip Q_3$};}
	
	\end{tikzpicture}}
	
	\vspacer

	\subcaptionbox{An affine prefix set $\Aa$ of a string $T$ with representation $\angles{X, (Q_1, 1, 2) \cdot (Q_2, 1, 3)}$ (drawn in black). If this representation is strongly affine, then its expansion $\angles{X, (Q_1, 1, 7) \cdot (Q_2, 1, 8)}$ is also a representation of an affine prefix set of $T$ (drawn in gray).\label{fig:strong}}[\textwidth]{
	\def\makestrongset{\affinesetup{%
		3/0/0/black/$X$,
		9/1/2/black/$Q_1$,
		1/1/3/black/{}%
		}%
	}
	
	\def\makeweaksetlong{\affinesetup[(weaklongstart)]{%
		0/0/0/gray/{},
		9/1/5/gray/$Q_1$,
		1/1/8/gray/{}%
		}%
	}

	\def\makeweaksetshort{\affinesetup[(weakshortstart)]{%
		0/0/0/gray/{},
		1/1/5/gray/{}%
		}%
	}

    \begin{tikzpicture}[x=.485em, y=1em]

	%\draw (affineq4-2-center) node[below, gray] {$\enskip Q_3$};
	
	{\tikzset{every path/.style={line width=1pt}}
	\makestrongset
	\affinedrawing[{\draw[densely dotted] (affinemax) to ++(0, -1.5em);}]}
	\node[above=0 of affinemax1 |- affinemax1-center] (q2) {$Q_2$};
	\node[inner sep=2pt] (q2) at (q2) {\phantom{$Q_2$}};
	\draw[thin, Latex-] (affinemax2-center) to[bend right] (q2);
	
	\node (weaklongstart) at (affinemax1) {};
	\node (weakshortstart) at (affinemax2) {};
	
	{\tikzset{every path/.style={thin}}
	\makeweaksetshort
	\affinedrawing[{\draw[densely dotted, gray] (affinemax) to ++(0, -1.5em);}]
	\path (weakshortstart) ++(-9,0) node (weakshortstart) {};
	\affinedrawing[{\draw[densely dotted, gray] (affinemax) to ++(0, -1.5em);}]
	
	\makeweaksetlong
	\affinedrawing[{\draw[densely dotted, gray] (affinemax) to ++(0, -1.5em);}]
	\draw (affinemax2-center) node[above, gray] {$Q_2$};}
	
	\draw (0,0) ++(0, -.5em) node (bottomleft) {};
	\draw (affinemax |- bottomleft) ++(1em, -1em) node (topright) {};
	\node[draw, fit=(bottomleft.center)(topright.center), inner sep=0] (text) {};
	\node[left=0 of text] {$T =$};
	
	\end{tikzpicture}}
	\caption{Affine prefix sets.}
\end{figure}

In what follows, we consider \emph{affine prefix sets}, i.e., affine sets that contain only prefixes of the given input string $T$.
We will show that a small number of affine prefix sets suffices to represent the $k$-palindromic prefixes of $T$.
An example for $k = 2$ is provided in \cref{fig:pal2}.
The case where $k = 1$, i.e., the structure of prefix-palindromes, is well-understood:
there are $\Oh(\log n)$ groups of such palindromes, where each group can be expressed as an arithmetic progression and a corresponding periodic prefix of $T$ (see, e.g., \cite[Lemma 5]{borozdin2017linear}).
Below, we restate this result in the framework of affine prefix sets (with proofs provided merely for self-containedness).

\begin{lemma}\label{lemma:prefix-pal1}
    Let $S$ be a palindrome, and let $P$ be the longest proper prefix of $S$ that is a palindrome. If $\absolute{S} \le 3\absolute{P}/2$,
    then there exist strings $U\in \pal\cup \{\emptystring\}$, $V\in \pal$ and an integer $i \ge 3$
    such that $P = U(VU)^{i - 1}$ and $S = U(VU)^{i}$. Furthermore, $VU$ is primitive, and $\absolute{VU}$ is the minimal period of both $S$ and $P$.
\end{lemma}
\begin{proof}
    Let $Q'$ be such that $S = PQ'$, and let $Q$ be the primitive root of $Q'$, i.e., $Q' = Q^h$ for some integer $h$. 
    As $\absolute{S} \le 3\absolute{P}/2$, we have $\absolute{Q'} \le \absolute{P}/2$.
    Now, since $P$ is a prefix-palindrome of $S$ and $S$ is a palindrome, $\rev{P} = P$ is a suffix of $S$. In other words, $S$ has a border $\absolute{P}$,
    and $\absolute{Q'}$ is a period of $S$.
    Since $S$ has suffix $Q'$, the periodicity implies that $S$ is a suffix of $Q'^{i'} = Q^{hi'}$ for some positive integer $i'$.
    Therefore, it holds $S = UQ^{i}$ for some integer $i \ge h$ and a proper suffix $U$ of $Q$.
    This implies $P = UQ^{i - h}$ because of $S = PQ^h$.
    
	Let $V$ be such that $Q = VU$, then $S$ has a prefix $UV$ and a suffix $VU$. Since $S$ is a palindrome, it holds $UV = \rev{VU} = \rev{U}\rev{V}$. We know that $V$ is non-empty because $U$ is a proper suffix of $Q$, thus $V \in\pal$ and $U \in\pal\cup\{\emptystring\}$.

    Let $j = i - h + 1$. Let $P' = U(VU)^{j} = (UV)^{j}U = (\rev{U}\rev{V})^{j}\rev{U} = \rev{U(VU)^{j}}$ and note that $P'$ is a palindrome of length $\absolute{PQ}$. If $h > 1$, i.e., if $Q'$ is non-primitive, then $P'$ is a proper prefix of $S$, which contradicts the fact that $P$ is the longest palindromic prefix of $S$. Hence $h = 1$ and $Q' = Q = VU$ is primitive. 
	If $i \le 2$ then $\absolute{Q} \ge \absolute{P}/2$, but we have already shown $\absolute{Q'} \le \absolute{P}/2$. Hence $i \ge 3$.
    Finally, both $P$ and $S$ have suffix $Q^2$ and period $\absolute{Q}$. By \cref{lem:primitive_squares}, $\absolute{Q}$ is the minimal period of $Q^2$, and thus also of $P$ and $S$.
\end{proof}

\begin{corollary}%
\label{coro:affine-set-pal1}%
    The prefix-palindromes of a string $T[1\dd n]$ can be partitioned into $\Oh(\log n)$ affine sets of order at most $1$. Each set of order $1$ has representation ${\angles{U(VU)^\ell, (VU, 1, u)}}$ for some $U\in \pal\cup \{\emptystring\}$, $V\in \pal$ and integers $\ell \ge 1$ and $u > 1$.
\end{corollary}
\begin{proof}
    Let $P_1, \ldots, P_m$ denote the prefix-palindromes of $T$, ordered by increasing lengths.
    We iterate over the prefix-palindromes starting with $P_2$, maintaining an active affine set $\Aa$ that is initially represented by $\angles{T[1], \emptyseries}$ of order 0 (because every string $T$ has shortest prefix-palindrome $T[1]$).
    %We maintain the invariant that, at the time of processing $P_i$, the active affine set either is either of order $0$ and contains only $P_{i - 1}$, or it is represented by $\angles{U(VU)^{\ell}, (VU, 1, u)}$ for some strings $U\in \pal\cup \{\emptystring\}$, $V\in \pal$ and positive integers $\ell$, $u$ such that $P_{i - 1} = U(VU)^{\ell + u}$.
    When processing $P_i$, we consider two cases:
    \begin{enumerate}[(a)]
        \item\label{case:new-affine-set} If $\absolute{P_i} > 3\absolute{P_{i-1}}/2$, we create a new affine set $\Aa$ containing only $P_i$, represented by $\angles{P_i, \emptyseries}$ of order $0$. This set will become the new active affine set. 
        \item\label{case:add-pref-to-active} If $\absolute{P_i} \le 3\absolute{P_{i-1}}/2$, then we add $P_i$ to the active set using one of two subcases.
        \begin{enumerate}[(i)]
        \vspace{.1\baselineskip}
        \item \label{case:add-pref-to-active:order0} If the representation of the active set is $\angles{P_{i - 1}, \emptyseries}$, then we add $P_i$ to the active set. By \cref{lemma:prefix-pal1}, there are $U\in \pal\cup \{\emptystring\}$, $V\in \pal$ and integer $j \ge 3$
    	such that $P_{i - 1} = U(VU)^{j - 1}$ and $P_i = U(VU)^{j}$. Furthermore, $VU$ is primitive, and $\absolute{VU}$ is the minimal period of $P_{i - 1}$ and $P_i$.
    	We replace $\angles{P_i, \emptyseries}$ with $\angles{U(VU)^{j - 2}, (VU, 1, 2)}$.
        
        \vspace{.1\baselineskip}
        \item\label{case:add-pref-to-active:order1} Otherwise, we can inductively assume the following invariant. The active set has representation $\angles{U(VU)^{\ell}, (VU, 1, u)}$ for integers $\ell \ge 1$ and $u > 1$, where $\absolute{VU}$ is the minimal period of the most recently added palindrome $P_{i - 1} = U(VU)^{\ell + u}$. Clearly, the invariant holds if $P_{i - 1}$ was processed using Case~(b.i), and we will ensure that it also holds after using Case~(b.ii).
        
        By \cref{lemma:prefix-pal1}, $P_i$ and $P_{i - 1}$ have the same minimal period, which is $\absolute{VU}$ due to the invariant. Also, the lemma implies ${P_{i}} = P_{i - 1}[1\dd \absolute{VU}]{P_{i - 1}}$, which is $U(VU)^{\ell + u + 1}$ due to the invariant. Hence we can add $P_i$ to the active set by merely increasing $u$ in the representation by one, which clearly maintains the invariant.
        \end{enumerate}
    \end{enumerate}
    We then proceed with the next palindrome $P_{i+1}$.
    In Case~(\ref{case:new-affine-set}), the length of the current prefix-palindrome exceeds the length of the previous one by a factor of $3/2$. Hence this case occurs at most $\ceil{\log_{3/2} n}$ times, resulting in $\Oh(\log n)$ affine sets. In Case~(\ref{case:add-pref-to-active}), we do not create any affine sets. The correctness follows from the description of the cases. In particular, as seen in Case~(\ref{case:add-pref-to-active}), the representations of order $1$ satisfy the stated properties.
\end{proof}

\subsection{Reducing affine prefix sets}

A single affine set may have multiple equivalent representations. For example, the affine set $S = \{\texttt{caba}, \texttt{cababa}, \texttt{cabababa}\}$ is represented by $\angles{\texttt c,(\texttt{ab}, 1, 3), (\texttt a, 1, 1)}$ and $\angles{\texttt{ca},(\texttt{ba}, 1, 3)}$.
Arguably, the latter representation is preferable, as it has a lower order and can thus be encoded more efficiently.
Hence we propose a way of potentially decreasing the order of a representation by \emph{reducing} it.

\begin{restatable}[Irreducible representation]{definition}{irreduciblereprdef}
A representation $\angles{X,\left(Q_i, \ell_i, u_i\right)_{i=1}^t}$ of an affine string set is \emph{irreducible} if and only if $\forall r \in [1, t] : 1 = \ell_r < u_r$ and $\forall r \in [1, t) : \absolute{Q_r} > \absolute{Q_{r + 1}}$.
\end{restatable}

From now on, we say that $Q_r$ with $r \in [1, t]$ is \emph{fixed} if $\ell_r = u_r$, and \emph{flexible} otherwise. If there is some $r \in [1, t)$ such that $\absolute{Q_r} \le \absolute{Q_{r + 1}}$, then we say that there is an \emph{inversion} between $Q_r$ and $Q_{r + 1}$. Thus, a representation is irreducible if and only if all components are flexible and have unit lower bounds, and there are no inversions.
As per this definition, $\angles{\texttt{ca},(\texttt{ba}, 1, 3)}$ is the only irreducible representation of $S$ from the previous example.
Another example is provided in \cref{fig:aps}.

\subparagraph{Properties of flexible components.}
Now we show how to make an arbitrary representation irreducible, possibly decreasing (but never increasing) its order. 
The reduction exploits the structure of periodic substrings induced by flexible components.

\begin{lemma}\label{lem:aux:structureflexperiod}
	Let $\angles{X,\left(Q_i, \ell_i, u_i\right)_{i=1}^t}$ be a representation of an affine prefix set, and consider any $r \in [1, t)$ such that $Q_r$ is flexible. Then $\absolute{Q_r}$ is a period of every string $Q_r^{a_{r}}Q_{r + 1}^{a_{r + 1}}\dots Q_{t}^{a_{t}}$ that satisfies $a_r \in \mathbb N_0$ and $\forall j \in (r, t] : a_j \in [\ell_j, u_j]$.
\end{lemma}
\begin{proof}
Let $P = Q_1^{\ell_1}Q_2^{\ell_2}\dots Q_{r}^{\ell_r}$ and $S = Q_{r + 1}^{a_{r + 1}}Q_{r + 2}^{a_{r + 2}}\dots Q_t^{a_t}$.
By the definition of an affine prefix set, $XPS$ is a prefix of the underlying string $T$. 
Since $Q_r$ is flexible, it holds $\ell_r < u_r$, and thus $XPQ_rS$ is also a prefix of $T$.
If both $XPS$ and $XPQ_rS$ are prefixes of $T$, then $S$ is a prefix of $Q_rS$. Hence $Q_rS$ has border $S$ and period $\absolute{Q_r}$. If $Q_rS$ has period $\absolute{Q_r}$, then clearly $Q_r^aS$, for all $a \in \mathbb{Z}^+$, also has period $\absolute{Q_r}$.
\end{proof}

If two adjacent components $Q_r$ and $Q_{r + 1}$ are flexible, then the lemma allows us to obtain the following lower bound on the length of $Q_r$.

\begin{lemma}%
\label{lem:aux:structureflextwopow}
\label{lem:aux:structureflexflex}%
	Let $\angles{X,\left(Q_i, \ell_i, u_i\right)_{i=1}^t}$ be a representation of an affine prefix set, and let $r \in [1, t)$. If both $Q_r$ and $Q_{r+1}$ are flexible, then either $Q_r = Q_{r + 1}$ or \[\absolute{Q_r} > \absolute{Q_{r + 1}^{u_{r + 1} - 1}} + \left( \sum_{j=r + 2}^t \absolute{Q_j^{u_j}} \right) + \gcd(\absolute{Q_r}, \absolute{Q_{r + 1}}).\]
\end{lemma}
\begin{proof}

\def\qplus{{Q'_{r + 1}}}
For flexible $Q_r$ and $Q_{r+1}$, let $q_r = \absolute{Q_r}$, $q_{r + 1} = \absolute{Q_{r + 1}}$, and $p = \gcd(q_r, q_{r + 1})$. Let ${\qplus} = Q_{r + 1}^{u_{r + 1}}Q_{r + 2}^{u_{r + 2}}\dots Q_t^{u_t}$.
By \cref{lem:aux:structureflexperiod}, both $q_r$ and $q_{r + 1}$ are periods of ${\qplus}$, and $q_r$ is a period of $Q_r{\qplus}$.
Since $q_r$ is a period of $Q_r\qplus$, it is also a period of $Q_r{Q_{r + 1}}$. Hence $Q_r = {Q_{r + 1}}$ if and only if $q_r = {q_{r + 1}}$. For the sake of contradiction, assume that the lemma does not hold, i.e., $q_r \neq q_{r + 1}$ and $q_r \le \absolute{{\qplus}} - q_{r+1} + p$. We make two observations.

First, ${\qplus}$ is of length $\absolute{{\qplus}} \ge q_r + q_{r + 1} - p$, and it has distinct periods $q_r$ and $q_{r+1}$. The periodicity lemma (\cref{lem:finewilf}) implies that $p$ is a period of ${\qplus}$, and thus also a period of its prefix $Q_{r + 1}$. If $p < q_{r + 1}$, then $Q_{r + 1} = Q_{r + 1}[1\dd p]^{q_{r + 1} / p}$, which contradicts the primitivity of $Q_{r + 1}$.
Second, ${\qplus}$ is of length $\absolute{{\qplus}} \ge q_r + q_{r + 1} - p \ge q_r$. 
Since $q_r$ is a period of $Q_r{\qplus}$, we know that $Q_r$ is a prefix of ${\qplus}$. 
Hence $p$ is also a period of $Q_r$. 
If $p < q_r$, then $Q_r = Q_{r}[1\dd p]^{q_r/p}$, which contradicts the primitivity of $Q_r$.

We have shown that $\gcd(q_r, q_{r + 1}) \ge \max(q_r, q_{r + 1})$. This is only possible if $\gcd(q_r, q_{r + 1}) = q_r = q_{r + 1}$, which contradicts the assumption that $q_r \neq q_{r + 1}$. Therefore, the lemma holds.
\end{proof}

\begin{lemma}%
\label{lem:aux:irreduciblecardinality}%
\label{lem:irreducible_maxorder}%
	Let $\angles{X,\left(Q_i, \ell_i, u_i\right)_{i=1}^t}$ be an irreducible representation of an affine prefix set $\Aa$ of a string of length $n$. Then it holds $\absolute{\Aa} = \prod_{i=1}^t u_i$ and $t \leq \log_2 n$.
\end{lemma}

\begin{proof}
\def\aprime{b}
	Let $E = \{ (a_i)_{i = 1}^t \mid \forall i \in [1, t] : a_i \in [1, u_i] \}$ of cardinality $\absolute{E} = \prod_{i=1}^t u_i$ be the set of exponent configurations admitted by the representation (where $[1, u_i] = [\ell_i, u_i]$ because the representation is irreducible). Then $\Aa = \{ XQ_1^{a_1}Q_2^{a_2}\dots Q_t^{a_t} \mid (a_i)_{i = 1}^t \in E \}$. In order to show $\absolute{\Aa} = \absolute{E}$, it suffices to show that no two elements in $E$ generate the same string. 
	
For the sake of contradiction, assume that there are distinct sequences $(a_i)_{i = 1}^t,({\aprime}_i)_{i = 1}^t \in E$ that generate the same string $S = XQ_1^{a_1}Q_2^{a_2}\dots Q_t^{a_t} = XQ_1^{{\aprime}_1}Q_2^{{\aprime}_2}\dots Q_t^{{\aprime}_t}$. 
Let $r \in [1, t]$ be the minimal index such that $a_r \neq {\aprime}_{r }$, and assume w.l.o.g.\ that $a_r > {\aprime}_r$. 
Then $S$ has the prefix $XQ_1^{a_1}\dots Q_{r-1}^{a_{r - 1}}Q_r^{{\aprime}_{r}} = XQ_1^{{\aprime}_1}\dots Q_{r-1}^{{\aprime}_{r - 1}}Q_r^{{\aprime}_{r}}$, and we can factorize the corresponding suffix in two different ways as $Q_r^{a_r - {\aprime}_r}Q_{r + 1}^{a_{r + 1}}\dots Q_t^{a_t} = Q_{r + 1}^{{\aprime}_{r + 1}}\dots Q_t^{{\aprime}_t}$. 
However, the two factorizations have different lengths $\absolute{Q_r^{a_r - {\aprime}_r}Q_{r + 1}^{a_{r + 1}}\dots Q_t^{a_t}} > \absolute{Q_rQ_{r + 1}} > \absolute{Q_{r + 1}^{{u}_{r + 1}}\dots Q_t^{{u}_t}} \ge \absolute{Q_{r + 1}^{{\aprime}_{r + 1}}\dots Q_t^{{\aprime}_t}}$, where the second inequality is due to \cref{lem:aux:structureflextwopow}. Because of this contradiction, there cannot be distinct sequences $(a_i)_{i = 1}^t,({\aprime}_i)_{i = 1}^t \in E$ that define the same string.

Finally, it holds $\forall i \in [1, t] : u_i \geq 2$ for any irreducible representation, and thus $\absolute{\Aa} = \prod_{i=1}^t u_i \geq 2^t$. Since trivially $\absolute{\Aa} \leq n$, it follows $2^t \leq n$ or equivalently $t \leq \log_2 n$.
\end{proof}

\subparagraph{Transforming representations.}
Now we use the properties of flexible components to transform an arbitrary representation into an irreducible one. We use the following operations.

\begin{lemma}\label{lem:aux:transform}
	Let $\rho = \angles{X,\left(Q_i, \ell_i, u_i\right)_{i=1}^t}$ be a representation of an affine prefix set.%
	
	\begin{enumerate}\newcommand{\apsop}[2][\vspace{.75\baselineskip}]{
		
		\vspace{.5\baselineskip}
		
		$#2$
		
		#1
	}
		\item\label{itemswitch} If there is $r \in [1, t)$ such that $Q_r$ is flexible and $Q_{r + 1}$ is fixed, then let $y = \absolute{Q_{r+1}^{\ell_{r + 1}}} \bmod \absolute{Q_r}$. The affine prefix set has representation
		
	\apsop{\textop{switch}_r(\rho) = \angles{X, (Q_i, \ell_i, u_i)_{i = 1}^{r - 1} \cdot (Q_{r + 1}, \ell_{r + 1}, u_{r + 1}) \cdot (\rot^y(Q_{r}), \ell_r, u_r) \cdot (Q_i, \ell_i, u_i)_{i = r + 2}^{t}}.}
	
		\item\label{itemmerge} If there is $r \in [1, t)$ such that such that both $Q_r$ and $Q_{r + 1}$ are flexible and $\absolute{Q_r} \le \absolute{Q_{r + 1}}$, then $Q_r = Q_{r + 1}$ and the affine prefix set has representation
		
	\apsop{\textop{merge}_r(\rho) = \angles{X, (Q_i, \ell_i, u_i)_{i = 1}^{r - 1} \cdot (Q_{r}, \ell_{r} + \ell_{r + 1}, u_r + u_{r + 1}) \cdot (Q_i, \ell_i, u_i)_{i = r + 2}^{t}}.}
	
		\item\label{itemsplit} If there is $r \in [1, t]$ such that $\ell_r > 1$, then the affine prefix set has representation
		
		\apsop{\textop{split}_r(\rho) = \angles{X, (Q_i, \ell_i, u_i)_{i = 1}^{r - 1} \cdot (Q_{r}, \ell_r - 1, \ell_r - 1) \cdot (Q_{r}, 1, u_r - \ell_r + 1) \cdot (Q_i, \ell_i, u_i)_{i = r + 1}^{t}}.}
		
		\item\label{itemtrunc} If $Q_1$ is fixed, then the affine prefix set has representation
		
		\apsop[]{\textop{truncate}(\rho) = \angles{XQ_1^{\ell_1}, (Q_{i + 1}, \ell_{i + 1}, u_{i + 1})_{i = 1}^{t - 1}}.}
	\end{enumerate}
\end{lemma}

\begin{proof}

Statements (\ref{itemsplit}) and (\ref{itemtrunc}) are trivial. For (\ref{itemmerge}), if $\absolute{Q_r} \le \absolute{Q_{r + 1}}$ and both $Q_r$ and $Q_{r + 1}$ are flexible, then \cref{lem:aux:structureflexflex} implies $Q_r = Q_{r + 1}$. Hence the statement follows.

Finally, we show that statement (\ref{itemswitch}) holds.
Assume that $Q_r$ is flexible and $Q_{r+1}$ is fixed.
Then \cref{lem:aux:structureflexperiod} implies that $\absolute{Q_{r}}$ is a period of $Q_{r}Q_{r + 1}^{\ell_{r + 1}}$, and thus $Q_{r + 1}^{\ell_{r + 1}} = Q_{r}^xQ_r[1\dd y]$ with $x = \floor{\absolute{Q_{r+1}^{\ell_{r + 1}}} / \absolute{Q_r}}$ and $y = \absolute{Q_{r+1}^{\ell_{r + 1}}} \bmod \absolute{Q_r}$. (Either $x$ or $y$ might be zero, but this is irrelevant for the proof.) Let $P = Q_r[1\dd y]$ and $S = Q_r(y\dd \absolute{Q_r}]$. Any rotation of a primitive string is primitive, and hence $\rot^y(Q_r) = SP$ is primitive.
For any exponent $a \in [\ell_r, u_r]$, it holds $Q_{r}^aQ_{r + 1}^{\ell_{r + 1}} = (PS)^a(PS)^xP = (PS)^xP(SP)^a = Q_{r + 1}^{\ell_{r + 1}}(\rot^y(Q_r))^a$. Hence the stated transformation does not change the represented affine prefix set.
\end{proof}

The leftmost (i.e., lowest index) fixed component $Q_r$ of a representation can either be removed with \textop{truncate} (if $r = 1$), or it can be moved further to the left with $\textop{switch}_{r - 1}$ (if $r > 1$).
By repeatedly applying \textop{truncate} and \textop{switch}, we obtain the following lemma.

\begin{lemma}\label{lem:aux:removefixed}
Let $\rho = \angles{X,\left(Q_i, \ell_i, u_i\right)_{i=1}^t}$ 
be a representation of an affine prefix set, and let $F = \{ j \in [1, t] \mid \ell_j < u_j\} = \{ j_1, \dots, j_{\absolute{F}}\}$ 
with $j_1 < j_2 < \dots < j_{\absolute{F}}$ be the indices of the flexible components. 
Then the affine prefix set has a representation $\angles{\hat X, (\hat Q_{j_i}, \ell_{j_i}, u_{j_i})_{i=1}^{\absolute{F}}}$ such that $\hat Q_{j_i}$ is a rotation of $Q_{j_i}$ for every $i \in [1, \absolute{F}]$.
Both $\hat X$ and all the $\hat Q_{j_i}$ are functions of $X$, $Q_1, \dots, Q_t$, and $\ell_1, \dots, \ell_t$, i.e., they are independent of $u_1, \dots, u_t$.
\end{lemma}

\begin{proof}
We transform $\rho$ by repeatedly applying \cref{lem:aux:transform}.
First, as long as there is some flexible component $Q_r$ that is followed by a fixed component $Q_{r + 1}$, we apply $\rho \gets \textop{switch}_r(\rho)$. 
Conceptually speaking, this moves fixed components further to the left and flexible components further to the right (without changing the order of the representation). 
Hence it is easy to see that the procedure terminates.
Afterwards, any component $Q_r$ is fixed if and only if $r \in [1, t - \absolute{F}]$.
Now we merge all the fixed components into $X$ by applying $\rho \gets \textop{truncate}(\rho)$ exactly $t - \absolute{F}$ times. This results in a representation of order $\absolute{F}$ (possibly~$0$) in which all components are flexible.

Whenever $\textop{switch}_r$ produces a rotation of one of the primitive strings, the outcome depends solely on the lengths and exponent lower bounds of the participating components, while the exponent upper bounds are irrelevant. Similarly, $\textop{truncate}$ modifies $X$ according to its current value, as well as the length and exponent lower bound of the current leftmost component. Hence $\hat X$ and all the $\hat Q_{j_i}$ are indeed independent of the exponent upper bounds.
\end{proof}

After applying \cref{lem:aux:removefixed}, we repeatedly apply \textop{merge} to remove all inversions. Then, we apply \textop{split} until all flexible components have exponent lower bound 1. This may result in new fixed components, which we remove with \cref{lem:aux:removefixed}, resulting in the following lemma.

\begin{lemma}\label{lem:aux:makeirreducible}
An affine prefix set represented by $%
%\rho = %
\angles{X,\left(Q_i, \ell_i, u_i\right)_{i=1}^t}$ has an irreducible representation of order $\absolute{L} \leq t$, where ${L} = {\{ \absolute{Q_r} \mid r \in [1, t] : \ell_r < u_r \}}$%
%is the set of distinct lengths of flexible components.
.
\end{lemma}

\begin{proof}
We obtain a representation $\rho' = \angles{X,\left(Q'_i, \ell'_i, u'_i\right)_{i=1}^{t'}}$ of order $t' = \absolute{\{ i \in [1, t] \mid \ell_i < u_i \}}$ that contains only flexible components using \cref{lem:aux:removefixed}.
Note that \cref{lem:aux:removefixed} implies that the components of this representation still have lengths in ${L} = {\{ \absolute{Q'_r} \mid r \in [1, t'] \}}$.
Consider any length $q \in {L}$, and let $r_{\min}, r_{\max} \in [1, t']$ be the respectively minimal and maximal index such that $\absolute{Q'_{r_{\min}}} = \absolute{Q'_{r_{\max}}} = q$.
Then \cref{lem:aux:structureflexflex} implies $\forall r \in [r_{\min}, r_{\max}] : Q'_{r} = Q'_{r_{\min}}$.
We apply $\rho' \gets \textop{merge}_{r_{\min}}(\rho')$ exactly $r_{\max} - r_{\min}$ times and merge all the components of length $q$ into a single one. By applying this procedure for each length $q \in L$, we obtain a representation of $\absolute{L}$ components.

Now $\rho'$ is inversion-free, contains only flexible components, and is of order $\absolute{L}$. As long as there is a flexible component $Q_r$ such that $\ell_r > 1$, we apply $\rho \gets \textop{split}_r(\rho)$. This creates at most one additional fixed component for each of the $\absolute{L}$ flexible components.
We remove the fixed components using \cref{lem:aux:removefixed}, leaving the number of flexible components, their exponent lower bounds, and their lengths unchanged. Hence we obtain the final irreducible representation of order $\absolute{L}$.
\end{proof}

\subsection{Strongly affine representations}
\label{sec:strongaffine}

Later, we will describe and exploit intricate properties of repetitive fragments induced by affine prefix sets. 
These properties are easier to show if we can assume that the periodicity can be extended beyond the considered region by a constant number of additional repetitions of the period.
Consequently, we introduce the notion of a \textit{strongly affine representation} of an affine prefix set of~$T$, which corresponds to representations in which the exponent upper bound of every (flexible) component can be increased by \strongconstverb{} and still yield an affine prefix set of $T$.
This is visualized in \cref{fig:strong}.

\begin{restatable}[Strongly affine representations]{definition}{strongaffinerepr}\label{def:stronglyperiodic}
    A representation $\rho = \angles{X, (Q_i, \ell_i, u_i)_{i=1}^t}$ of an affine prefix set of a string $T$ is \emph{strongly affine} if and only if its \emph{periodic expansion}
    \[\textop{expand}(\rho) = \angles{X, (Q_i, \ell_i, u_i')_{i=1}^t}\text{\ with\ }\forall i \in [1, t] : \begin{cases}
    u_i' = u_i &\text{if $u_i = \ell_i$}\\
    u_i' = u_i + \strongconst &\text{otherwise}
    \end{cases}\]
    is also the representation of an affine prefix set of $T$.
\end{restatable}

\begin{restatable}[Canonical representation]{definition}{canonicalrepr} A representation of an affine prefix set is \emph{canonical} if and only if it is both strongly affine and irreducible.
\end{restatable}

It can be readily verified that, if $\rho$ is strongly affine, then also $\textop{truncate}(\rho)$, $\textop{split}_r(\rho)$, $\textop{merge}_r(\rho)$, and $\textop{switch}_r(\rho)$ are strongly affine (for any $r$, assuming that the respective operation is indeed applicable). We obtained \cref{lem:aux:makeirreducible} by applying a sequence of these operations, and hence we have the following immediate corollary.

\begin{corollary}\label{lem:aux:strongtoirreducible}
An affine prefix set with strongly affine representation $%
%\rho = %
\angles{X,\left(Q_i, \ell_i, u_i\right)_{i=1}^t}$ has a canonical representation of order $\absolute{L} \leq t$, where ${L} = {\{ \absolute{Q_r} \mid r \in [1, t] : \ell_r < u_r \}}$%
%is the set of distinct lengths of flexible components.
.
\end{corollary}

Whether a representation $\rho$ of an affine prefix set $\Aa$ of $T$ is strongly affine does not only depend on $\rho$, it also depends on what $T$ looks like beyond the end of the longest prefix represented by $\rho$. Therefore, one cannot hope to transform an arbitrary representation into an equivalent strongly affine representation.
However, by ``removing'' the last \strongconstverb{} copies of each component and treating them separately, we show that we can cover an affine prefix set of order $t$ with at most $\strongconstplus^t$ canonical representations.

\begin{restatable}{lemma}{makestrong}\label{lem:aux:makestrong}
    An affine prefix set of order $t$ can be partitioned into $\strongconstplus^t$ affine prefix sets, each of which has a canonical representation of order at most $t$.
\end{restatable}
\begin{proof}
    Let $\angles{X, (Q_i, \ell_i, u_i)_{i=1}^t}$ be a representation of an affine prefix set. We produce a set of representations defined by 
   \begin{equation}
   \label{eq:strong} 
    \begin{split}
    R = &\left\{ \angles{X, (Q_i, \ell'_i, u'_i)_{i=1}^t} \mid \forall r \in [1, t] : (\ell_r', u_r') \in B_r \right\},\\
    \textnormal{where }\forall r \in [1, t] : B_r = %
    &\left\{ (u, u) \mid u \in [\max(1, u_r - 4), u_r] \right\} \cup \{ (\ell_r, \max(\ell_r, u_r - \strongconst)) \}.
    %&\left\{ (u_r, u_r), \ldots, (u_r - \strongconstminus, u_r - \strongconstminus), (\ell_r, \max(\ell_r, u_r - \strongconst)) \vphantom{\ell_r')_{i=1}^t}\right\}
    \end{split}%
    \end{equation}
    It is easy to see that the affine sets generated by representations in $R$ form a partition of the affine set generated by $\angles{X, (Q_i, \ell_i, u_i)_{i=1}^t}$.
    By design, for any representation in $R$, and for any component $Q_r$, we know that $Q_r$ is either fixed, or it has exponent lower bound $\ell_r$ and exponent upper bound $u_r - \strongconst$.
    Hence the instances in $R$ are strongly affine, and it follows from \cref{lem:aux:strongtoirreducible} that each of them has an equivalent canonical representation of order at most $t$.
    Finally, it holds $\forall i \in [1, t] : \absolute{B_r} \leq \strongconstplus$ and thus $\absolute{R} = \prod_{i = 1}^t \absolute{B_i} \leq \strongconstplus^t$.
\end{proof}

By applying the technique from the proof above to the prefix-palindromes, i.e., to each of the representations of order $1$ produced by \cref{coro:affine-set-pal1}, we obtain the following result.

\begin{corollary}\label{lm:base_case}
The set of prefix-palindromes of a string $T[1\dd n]$ can be partitioned into $\Oh(\log n)$ affine sets of order at most $1$. Each set of order $1$ has strongly affine representation $\angles{U(VU)^\ell, (VU, 1, u)}$ for some $U\in \pal\cup \{\emptystring\}$, $V\in \pal$ and integers $\ell \ge 1$ and $u > 1$.
\end{corollary}

\begin{corollary}\label{lem:aux:structure:stronglength}
Let $\angles{X, (Q_i, \ell_i, u_i)_{i=1}^t}$ be a canonical representation of an affine prefix set. Then it holds $\forall r \in [1, t] : \absolute{Q_r} > \sum_{j=r + 1}^t \absolute{Q_{j}^{u_j + \strongconstminus}}$.

\end{corollary}

\begin{proof}
If $\rho = \angles{X, (Q_i, \ell_i, u_i)_{i=1}^t}$ is canonical, then clearly $\textop{expand}(\rho) = \angles{X, (Q_i, \ell_i, u_i + \strongconst)_{i=1}^t}$ is irreducible. Thus, the statement follows from \cref{lem:aux:structureflextwopow} applied to $\textop{expand}(\rho)$.
\end{proof}

\begin{lemma}\label{lem:aux:structure:aps_of_q1}
    Let $\angles{X, (Q_i, 1, u_i)_{i=1}^t}$ be a canonical representation of an affine prefix set, and let $h \in [0, \strongconst]$. Then $\angles{\emptystring, (Q_i, 1, u_i + h)_{i=2}^t}$ is an irreducible representation of an affine prefix set of the string $Q_1^2$, and, if $h < \strongconst$, also of the string $Q_1$.
\end{lemma}

\begin{proof}
    Consider any $h \in [0, \strongconst]$. Due to the strong affinity, $\angles{X, (Q_i, 1, u_i + h)_{i=1}^t}$ represents an affine prefix set. Let $(a_i)_{i=2}^t$ be a sequence of exponents with $\forall j \in [2, t] : a_j \in [1, u_i + h]$. By \cref{lem:aux:structureflexperiod}, the string $Q_1Q_2^{a_2}Q_3^{a_3}\dots Q_t^{a_t}$ has period $Q_1$. Due to \cref{lem:aux:structureflextwopow}, it holds $\absolute{Q_2^{a_2}Q_3^{a_3}\dots Q_t^{a_t}} < \absolute{Q_1Q_2} < \absolute{Q_1^2}$. Hence we have shown that $Q_2^{a_2}Q_3^{a_3}\dots Q_t^{a_t}$ is a prefix of $Q_1^2$, and $\angles{\emptystring, (Q_i, 1, u_i + h)_{i=2}^t}$ is a representation of an affine prefix set of $Q_1^2$. Since $\angles{X, (Q_i, 1, u_i)_{i=1}^t}$ is irreducible, it is easy to see that also $\angles{\emptystring, (Q_i, 1, u_i + h)_{i=2}^t}$ is irreducible.
    
If $h < \strongconst$, then \cref{lem:aux:structureflextwopow} invoked with $\angles{X, (Q_i, 1, u_i + \strongconst)_{i=1}^t}$ implies $\absolute{Q_2^{a_2}Q_3^{a_3}\dots Q_t^{a_t}} < \absolute{Q_1}$, 
and $\angles{\emptystring, (Q_i, 1, u_i + h)_{i=2}^t}$ indeed only generates strings of length less than $\absolute{Q_1}$.
\end{proof}

\begin{corollary}\label{lem:aux:structure:aps_of_q1:strong}
    Let $\angles{X, (Q_i, 1, u_i)_{i=1}^t}$ be a canonical representation of an affine prefix set. Then $\angles{\emptystring, (Q_i, 1, u_i)_{i=2}^t}$ is a canonical representation of an affine prefix set of the string $Q_1^2$.
\end{corollary}

\begin{proof} By \cref{lem:aux:structure:aps_of_q1}, $\angles{\emptystring, (Q_i, 1, u_i + \strongconst)_{i=2}^t}$ is an irreducible representation of an affine prefix set of $Q_1^2$. Hence $\angles{\emptystring, (Q_i, 1, u_i)_{i=2}^t}$ is a canonical representation for $Q_1^2$.
\end{proof}

\subsection{Reversing the structure of affine prefix sets}\label{sec:reversal}
We first show that a periodic fragment of $T$ induced by an affine prefix set can be covered by a combination of a forward and a ``backward'' affine prefix set. This is formally expressed by the lemma below, and visualized in \cref{fig:apsrev}.

\def\smax{S}
\def\hatQ{{\hat Q}}
\begin{lemma}\label{lem:aux:structure:factorizesmax}
Let $\angles{X, (Q_i, 1, u_i)_{i=1}^t}$ be an irreducible representation of an affine prefix set, let $\smax = Q_1^{u_1}Q_{2}^{u_{2}}\dots Q_t^{u_t}$, and for $j \in [1, t]$ let $\hatQ_j$ be the length-$\absolute{Q_j}$ suffix of $\smax$. For any sequence $(a_i)_{i=1}^t$ with $\forall j \in [1, t] : a_j \in [0, u_j]$, it holds
$$\smax = Q_1^{u_1 - a_1}Q_{2}^{u_2 - a_{2}}\dots Q_t^{u_t - a_t}\enskip \cdot \enskip\hatQ_t^{a_t}\hatQ_{t - 1}^{a_{t - 1}}\dots \hatQ_{1}^{a_1}.$$
\end{lemma}

\begin{proof}
If $t = 1$, then $\smax = Q_1^{u_1} = \hatQ_1^{u_1} = Q_1^{u_1-a_1}\hatQ_1^{a_1}$.
Inductively assume that the lemma holds for representations of order $t - 1$. Now we show that it holds for representations of order $t$.
If $\angles{X, (Q_i, 1, u_i)_{i=1}^t}$ is an irreducible representation of an affine prefix set, then clearly $\angles{XQ_1^{u_1}, (Q_i, 1, u_i)_{i=2}^t}$ is an irreducible representation of another affine prefix set. This representation is of order $t-1$, and hence the inductive assumption implies
\[\smax = Q_1^{u_1} \enskip\cdot\enskip Q_2^{u_2 - a_2}Q_{3}^{u_3 - a_{3}}\dots Q_t^{u_t - a_t}\enskip \cdot \enskip\hatQ_t^{a_t}\hatQ_{t - 1}^{a_{t - 1}}\dots \hatQ_{2}^{a_2}.\]
If $a_1 = 0$, then there is nothing left to do. Hence assume $a_1 > 0$.
Since $\angles{X, (Q_i, 1, u_i)_{i=1}^t}$ is an irreducible representation, \cref{lem:aux:structureflexperiod} implies that $\absolute{Q_1}$ and therefore also $q = a_1 \cdot \absolute{Q_1}$ is a period of $\smax$. Hence $\smax$ has a border of length $s - q$, where $s = \absolute{S}$, and it holds
\[\smax[1\dd s-q] = \smax[1+q\dd s] = Q_1^{u_1 - a_1} \enskip\cdot\enskip Q_2^{u_2 - a_2}Q_{3}^{u_3 - a_{3}}\dots Q_t^{u_t - a_t}\enskip \cdot \enskip\hatQ_t^{a_t}\hatQ_{t - 1}^{a_{t - 1}}\dots \hatQ_{2}^{a_2}.\]
Finally, as mentioned before, $\smax[s-q+1\dd s]$ of length $q = a_1 \cdot \absolute{Q_1}$ has period $\absolute{Q_1}$. Hence $\smax[s-q+1\dd s] = (\smax[s-\absolute{Q_1}+1\dd s])^{a_1} = \hatQ_1^{a_1}$, which concludes the proof.
\end{proof}

\begin{figure}
    \edef\q{10.5}
    \edef\qq{2.5}
    \edef\qqq{1}

    \edef\u{2}
    \edef\uu{3}
    \edef\uuu{2}

    \edef\yoff{2.5}
    \edef\yyoff{2.5}
    \edef\yyyoff{1.5}

    \providecommand\theanyQ[2]{$\scriptstyle\vphantom{\hat{Q}}{#1}_{#2}$}
    \providecommand\theQ[1]{\theanyQ{Q}{#1}}
    \providecommand\theunflippedQ[1]{\theanyQ{\hat Q}{#1}}
    \providecommand\theflippedQ[1]{\rotatebox{180}{\theunflippedQ{#1}}}

    \def\flipdrawing{%
    \xdef\yoff{-\yoff}%
    \xdef\yyoff{-\yyoff}%
    \xdef\yyyoff{-\yyyoff}}

    \edef\smalloffset{0.05}

    \providecommand\theedge[5]{%
        \edef\src{#1}
        \pgfmathparse{\src+#2}
        \edef\dst{\pgfmathresult}

        \edef\yoff{#3}
        \def\lbl{#4}

        \node at (\src, 0) {};
        \node at (\dst, 0) {};
      
        \path (\src, 0) ++(\smalloffset, 0) node (src) {};
        \path (\dst, 0) ++(-\smalloffset, 0) node (dst) {};
        \path (\src, 0) ++(0, \yoff) node (v) {};
        \path[#5, draw=none] (src.center) to (src |- v) to[in=90, out=90] node[midway,below,fill=white,inner sep=0,outer sep=.3em] {\lbl} (dst |- v) to (dst.center);
        \draw[#5] (src.center) to (src |- v) to[in=90, out=90] (dst |- v) to (dst.center);
    }

    \providecommand\qedge[2]{%
        \ifnum#1<\u\pgfmathparse{#1*\q}\edef\xoffset{\pgfmathresult}
        \theedge{\xoffset}{\q}{\yoff}{\theQ{1}}{#2}\fi
    }
    \providecommand\qqedge[3]{%
        \ifnum#2<\uu\pgfmathparse{#2*\qq+#1*\q}\edef\xoffset{\pgfmathresult}
        \theedge{\xoffset}{\qq}{\yyoff}{\theQ{2}}{#3}\fi
    }
    \providecommand\qqqedge[4]{%
        \ifnum#3<\uuu\pgfmathparse{#3*\qqq+#2*\qq+#1*\q}\edef\xoffset{\pgfmathresult}
        \theedge{\xoffset}{\qqq}{\yyyoff}{\theQ{3}}{#4}\fi
    }

    \providecommand\thestructureparamall[6]{
    \foreach \a in {#1} { \qedge{\a}{#4}
        \foreach \aa in {#2} { \qqedge{\a}{\aa}{#5}
            \foreach \aaa in {#3} { \qqqedge{\a}{\aa}{\aaa}{#6}}}}}

    \providecommand\thestructureparam[3]{\thestructureparamall{0,...,\u}{0,...,\uu}{0,...,\uuu}{#1}{#2}{#3}}
    \providecommand\thestructure[1][gray,densely dotted]{\thestructureparam{#1}{#1}{#1}}

    \providecommand\thepathparam[6]{
    \ifnum#1>0\foreach[count=\a from 0] \i in {1,...,#1} { \qedge{\a}{#4} }\fi
    \ifnum#2>0\foreach[count=\aa from 0] \i in {1,...,#2} { \qqedge{#1}{\aa}{#5} }\fi
    \ifnum#3>0\foreach[count=\aaa from 0] \i in {1,...,#3} { \qqqedge{#1}{#2}{\aaa}{#6} }\fi}
        
    \providecommand\thepath[4][thick,black]{\thepathparam{#2}{#3}{#4}{#1}{#1}{#1}}

    \tikzset{theframe/.style={x=1.27em, y=1em}}
    \centering
    \scalebox{.95}{\begin{tikzpicture}[theframe]
        \thestructure
        \thepath{1}{2}{1}

        \def\prevpos{0}
        \def\pos{0}
        \foreach[count=\i from 0, remember=\pos as \prevpos initially 0, evaluate=\len as \pos using (\prevpos+\len)] \len in {\q,\qq,\qq,\qqq,\qqq,\qq,\q} {
            \draw (\pos, -.5) rectangle node[midway] (lbl\i) {} (\prevpos, -1.75);
        }
        \foreach[count=\i from 0] \lbl in {\theQ{1},\theQ{2},\theQ{2},\theQ{3},\theunflippedQ{3},\theunflippedQ{2},\theunflippedQ{1}} {
            \node at (lbl\i) {\lbl};
        }
        \node (leftborder) at (0, 0) {};
        \node[left=0 of leftborder |- lbl0] {{$S =\ $}};
    \end{tikzpicture}}

    \vspace{-.9\baselineskip}

    \let\theQ\theflippedQ%
    \scalebox{.95}{\rotatebox{180}{%
    \begin{tikzpicture}[theframe]
        \thestructure
        \thepath{1}{1}{1}
        
        \node (leftborder) at (30.25, 2) {};
        \node[right=0 of leftborder] {\phantom{$S =\ $}};
    \end{tikzpicture}}}
    \caption{\cref{lem:aux:structure:factorizesmax} applied to an irreducible representation $\angles{X, (Q_1, 1, 2) \cdot (Q_2, 1, 3) \cdot (Q_3, 1, 2)}$. The drawing shows the longest prefix $S = Q_1^{2} Q_2^{3} Q_3^{2}$ generated by the representation. By the lemma, for any $a_1 \in [0,2]$, $a_2 \in [0,3]$ and $a_3 \in [0,2]$, it holds $S = Q_1^{2 - a_1} Q_2^{3-a_2} Q_3^{2-a_3} \ \cdot\ \hatQ_3^{a_3}\hatQ_{2}^{a_{2}}\hatQ_{1}^{a_1}$, where each $\hatQ_j$ is the length-$\absolute{Q_j}$ suffix of $S$. The drawing highlights the case where $a_1 = a_2 = a_3 = 1$.}
    \label{fig:apsrev}
\end{figure}
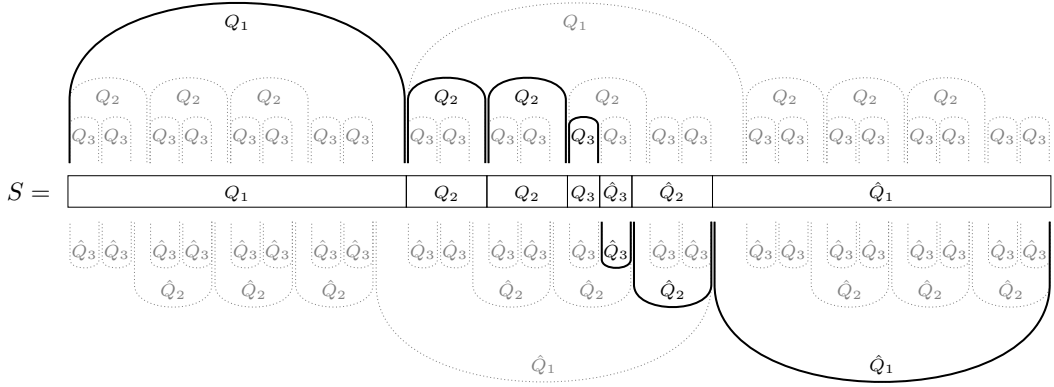

We now build on this characterization to convert irreducible representations of affine prefix sets of $S$ into irreducible representations of affine prefix sets of $\rev{S}$.

\begin{corollary}\label{lem:aux:structure:reverse_aps}
Let $\angles{X, (Q_i, 1, u_i)_{i=1}^t}$ be a canonical representation of an affine prefix set, let $s = \sum_{i = 2}^t (u_i + 1) \cdot \absolute{Q_i}$, and for $j \in [1, t]$ let $\hatQ_j$ be the length-$\absolute{Q_j}$ suffix of $\rot^s(Q_1)$. 
Then $\angles{\emptystring, (\rev{\hatQ_i}, 1, u_i)_{i = 2}^t}$ represents an affine prefix set of $\rev{\rot^s(Q_1)}$.
\end{corollary}

\begin{proof}
	\sloppy
    Consider any sequence $(a_i)_{i=2}^t$ of exponents admitted by the representation, i.e., $\forall j \in [2, t] : a_j \in [1, u_j]$.
    By \cref{lem:aux:structure:aps_of_q1}, $\angles{\emptystring, (Q_i, 1, u_i + 1)_{i=2}^t}$ is an irreducible representation of an affine prefix set of $Q_1$, which implies $Q_1[1\dd s] = Q_2^{u_2 + 1}Q_3^{u_3 + 1}\dots Q_t^{u_t + 1}$. For this representation, \cref{lem:aux:structure:factorizesmax} implies that $\hatQ_t^{a_t}\hatQ_{t - 1}^{a_{t - 1}}\dots \hatQ_{2}^{a_2}$ is a suffix of $Q_1[1\dd s]$. Thus, its reversal $\rev{\hatQ_t^{a_t}\hatQ_{t - 1}^{a_{t - 1}}\dots \hatQ_{2}^{a_2}} = \rev{\hatQ_{2}^{a_2}}\rev{\hatQ_{3}^{a_3}}\dots \rev{\hatQ_{t}^{a_t}}$ is a prefix of $\rev{Q_1[1\dd s]}$, which is a prefix of $\rev{\rot^s(Q_1)}$.
\end{proof}

\section{Appending a Palindrome to an Affine Prefix Set}
\label{sec:appendapalindrome}

In this section, we describe how to extend an affine prefix set $\Aa$ with a palindrome. This broadly means that we want to compute a union of multiple affine prefix sets, such that each of the new prefixes is the concatenation of a prefix in $\Aa$ and a palindrome. We distinguish two cases depending on whether or not the palindrome to be appended is inside a periodic fragment of $T$ or not. Regardless of the case, we may first overextend $\Aa$ so that the new affine sets are not necessarily affine prefix sets. Whenever this happens, we truncate the sets by restricting the length of their strings with the auxiliary lemma below.
For a set of strings $\Aa$, denote $\Aa |_m = \{S \in \Aa : \absolute{S} \le m\}$.

\begin{lemma}\label{lem:truncate_affine_set}
Let $\angles{X, (Q_i, \ell_i, u_i)_{i=1}^t}$ be a representation of an affine prefix set $\Aa$. For $m \in \mathbb N$, we can express $\Aa' = \Aa |_m$ as a union of at most $t' \le t$ affine prefix sets $\Aa' = \bigcup_{j = 1}^{t'} \Aa_j$, each with a representation of order at most $t$.
\end{lemma}

\begin{proof}
    Due to \cref{lem:aux:makeirreducible}, we can assume that the given representation is irreducible.
    If $\absolute{XQ_1} > m$, then $\Aa'$ is empty, and thus it is the union of zero affine prefix sets. If $\absolute{XQ_1} \le m$ and $t = 1$, then we simply create a single new representation $\angles{X, (Q_1, 1, \floor{(m - \absolute{X}) / \absolute{Q_1}})}$ that represents $\Aa'$. The proof for $t > 1$ works by induction. 
    
    Assume that the lemma holds for representations of order at most $t - 1$. Now we show that it holds for representations of order $t$.    
    Let $a_1 \in \mathbb N$ be the minimal exponent such that $a_1 \ge \ell_1$ and $\absolute{XQ_1^{a_1}Q_2^{u_2}Q_3^{u_3}\dots Q_t^{u_t}} > m$. If $a_1 > u_1$, then $\Aa = \Aa'$ and there is nothing left to do.
    If $a_1 \le u_1$, then \cref{lem:aux:structureflextwopow} (and the fact that the representation is irreducible) implies $\absolute{XQ_1^{a_1 + 1}Q_2} > \absolute{XQ_1^{a_1}Q_2^{u_2}Q_3^{u_3}\dots Q_t^{u_t}} > m$.
    Thus, we do not need to consider prefixes that are generated by using an exponent larger than $a_1$ for $Q_1$, and we partition the remaining prefixes into two affine prefix sets. 
    The first one is $\Aa''$ represented by $\angles{XQ_1^{a_1}, (Q_i, \ell_i, u_i)_{i=2}^t}$, i.e., the set that contains all the strings that use exponent $a_1$ for $Q_1$. Its representation is of order $t - 1$, and the inductive assumption implies that it is the union of $t' - 1 \le t - 1$ affine prefix sets $\Aa'' = \bigcup_{j = 1}^{t' - 1} \Aa_j$.
    For the remaining strings, if $a_1 > \ell_1$, we create one additional set $\Aa_{t'}$ represented by $\angles{X, (Q_1, \ell_1, a_1 - 1) \cdot (Q_i, \ell_i, u_i)_{i=2}^t}$. It then holds $\Aa' = \Aa'' \cup \Aa_{t'}$. If, however, $a_1 = \ell_1$, then it already holds $\Aa' = \Aa''$ and there is nothing left to do.
\end{proof}

\subsection{Appending a long palindrome}
Assume that the affine prefix set to be extended is given in a canonical representation $\angles{X, (Q_i, 1, u_i)_{i=1}^t}$. We first focus on appending \emph{long palindromes} of length at least $2\absolute{Q_1}$, and then we show that the shorter palindromes can be handled recursively. Note that, for a canonical representation, $T$ has a prefix $XQ_1^{u_1 + \strongconst}$. At the same time, the longest prefix in the affine set is of length less than $XQ_1^{u_1 + 1}$. This leads us to a case distinction based on the center of the palindrome to be appended. If the center is before position $\absolute{XQ_1^{u_1 + 3}}$, then we can show that the entire palindrome is within the $\absolute{Q_1}$-periodic prefix of $T[\absolute{X}+1\dd n]$. Otherwise, the left half of the palindrome contains position $\absolute{XQ_1^{u_1 + 2}}$, and we can use this position as an anchor point for the extension.

\subsubsection%
[Appending a long palindrome within a run of Q1]%
{Appending a long palindrome within a run of \boldmath$Q_1$\unboldmath}

We now focus on the case where the long palindrome to be appended is entirely within the $\absolute{Q_1}$-periodic prefix of $T[\absolute{X} + 1\dd n]$. We proceed in two steps. First (in \cref{th:inner}), we show how to append a palindrome under the assumption that the entire string has the form $XQ_1^x$ for some integer~$x$. The second step (\cref{cor:inner}) truncates the result of the first step such that it corresponds to $XQ_1^{\alpha}$, where $\alpha \in \mathbb Q$ is the largest value such that $XQ_1^{\alpha}$ is a prefix of $T$.

\begin{theorem}\label{th:inner}
    Let $\angles{X, (Q_i, 1, u_i)_{i=1}^t}$ be a canonical representation of an affine prefix set $\Aa$. Let $s = \sum_{i = 2}^t (u_i + 1) \cdot \absolute{Q_i}$, and for $j \in [1, t]$ let $\hatQ_j$ be the length-$\absolute{Q_j}$ suffix of $\rot^s(Q_1)$. If $\rot^r(Q_1) = \rev{Q_1}$ for some $r \in [s, s + \absolute{Q_1})$, then 
    \begin{equation}\label{eq:new_inside}
    \angles{X \cdot Q_1 \cdot Q_1[1\dd r - s], (\rot^{r - s}(Q_1), 1, x) \cdot (\rev{\hatQ_i}, 1, u_i)_{i = 2}^t} 
    \end{equation}    
    represents an affine prefix set $\Aa'$ of $XQ_1^{x + 3}$, for any positive integer $x$. Additionally, each of the following holds:

    \begin{enumerate}
        \item If $Y' \in \Aa'$, then there is a string $Y \in \Aa$ and a palindrome $P$ such that $Y' = YP$.
        \item For $Y \in \Aa$ and $P \in \pal$, if $\absolute{P} \ge 2\absolute{Q_1}$ and $YP$ is a prefix of $XQ_1^{x + 1}$, then $YP \in \Aa'$.
    \end{enumerate}
\end{theorem}

\begin{proof}
    Let $q = \absolute{Q_1}$. By \cref{lem:aux:structure:stronglength}, it holds $s < q$. Consider any string $XS' \in \Aa'$. Since this string is generated by \cref{eq:new_inside}, there must be exponents $a_1 \in [1, x]$ and $\forall j \in [2, t] : a_j \in [1, u_j]$ such that 
        $$S' = Q_1 \cdot Q_1[1\dd r - s] \cdot \rot^{r - s}(Q_1)^{a_1} \cdot Q' = Q_1^{a_1 + 1} \cdot Q_1[1\dd r - s] \cdot Q',$$
    where $Q' = \rev{\hatQ_2}^{a_2}\rev{\hatQ_3}^{a_3}\dots \rev{\hatQ_t}^{a_t}$.
    We start by showing that \cref{eq:new_inside} represents an affine prefix set of $XQ_1^{x + 3}$, i.e., we must show that $S'$ is a prefix of $Q_1^{x + 3}$.
    The suffix $Q'$ of $S'$ was generated by the last part $(\rev{\hatQ_i}, 1, u_i)_{i = 2}^t$ of \cref{eq:new_inside}. \cref{lem:aux:structure:reverse_aps} implies that $Q'$ is a prefix of $$\rev{\rot^s(Q_1)} = \rot^{-s}(\rev{Q_1}) = \rot^{-s}(\rot^r(Q_1)) = \rot^{r - s}(Q_1).$$ 
    Therefore, it holds 
    $$S' = Q_1^{a_1 + 1} \cdot Q_1[1\dd r - s] \cdot (\rot^{r - s}(Q_1))[1\dd \absolute{Q'}] = Q_1^{a_1 + 1} \cdot (Q_1^2)[1\dd r - s + \absolute{Q'}],$$ 
    and $S'$ is a prefix of $Q_1^{a_1 + 3}$, which is a prefix of $Q_1^{x + 3}$.

    Next, we show that $S' = SP$ for some string $S$ with $XS \in \Aa$ and a palindrome $P$.
    It holds $a_j \in [1, u_j]$ if and only if $u_j - a_j + 1 \in [1, u_j]$, and thus $\angles{X, (Q_i, 1, u_i)_{i=1}^t}$ generates the string $XS \in \Aa$ with $S = Q_1^1Q_2^{u_2 - a_2 + 1}Q_3^{u_3 - a_3 + 1} \dots Q_t^{u_t - a_t + 1}$, where $\absolute{S} = q + s - \absolute{Q'}$. 
    Let $P$ be the unique string such that $S' = SP$, i.e.,
    $$P = S'[q + s - \absolute{Q'} + 1\dd \absolute{S'}] = Q_1[s - \absolute{Q'} + 1\dd q] \cdot Q_1^{a_1 - 1} \cdot (Q_1^2)[1\dd r - s + \absolute{Q'}].$$

    It remains to be shown that $P$ is a palindrome. 
    Let $L = Q_1[s - \absolute{Q'} + 1\dd q]$ and $R = (Q_1^2)[1\dd r - s + \absolute{Q'}]$, then \cref{lem:aux:structure:palindrome_in_run:extended} implies that $P$ is a palindrome if $\rot^{\absolute{R} - \absolute{L}}(Q_1) = \rev{Q_1}$.
    Indeed, cyclically shifting $Q_1$ by
        $$\absolute{R} - \absolute{L} = (r - s + \absolute{Q'}) - (q - (s - \absolute{Q'} + 1) + 1) = r - q$$
    steps is equivalent to cyclically shifting it by $r$ steps, and thus $\rot^{\absolute{R} - \absolute{L}}(Q_1) = \rot^{r - q}(Q_1) = \rot^{r}(Q_1) = \rev{Q_1}$. Hence $P$ is a palindrome.

    Finally, consider any string $XS \in \Aa$ and a palindrome $P \in \pal$ of length $\absolute{P} \ge 2\absolute{Q_1}$ such that $SP$ is a prefix of $Q_1^{x + 1}$.
    Since $P$ is a substring of $Q_1^{x + 1}$, it must be of the form $LQ_1^zR$ for some positive integer $z$, a suffix $L$ of $Q_1$, and a prefix $R$ of $Q_1$.    
    Due to $XS \in \Aa$, there is some sequence $\forall j \in [1, t] : a_j \in [1, u_j]$ of exponents such that $S = Q_1^{u_1 - a_1 + 1}Q_2^{u_2 - a_2 + 1}\dots Q_t^{u_t - a_t + 1}$. Let $q' = \sum_{j = 2}^t a_j \cdot \absolute{Q_j}$, then $L = Q_1[s - q' + 1\dd q]$ and $\absolute{R} = \absolute{P} - zq - q + s - q'$. Since $P$ is a palindrome, \cref{lem:aux:structure:palindrome_in_run:extended} implies that $\rot^{\absolute{R} - \absolute{L}}(Q_1) = \rev{Q_1}$.
    Since $Q_1$ is primitive and $\rev{Q_1} = \rot^r({Q_1})$, it is necessary that $\absolute{R} - \absolute{L} = z' q + r$ for some integer $z'$.
    This leads to $$\absolute{R} - \absolute{L} = \absolute{P} + 2(s - q') - zq - 2q = z' q + r,$$ or equivalently
    $$\absolute{P} = (z' + z + 2)q + r - 2(s - q').$$
    We have shown that $SP$ is of length $\absolute{P} + \absolute{S} = (u_1 - a_1 + z' + z + 2)q + q + r - s + q'$ for some integers $z,z'$ such that $\absolute{SP} \le (x + 1)q$. 
    Let $x' = (u_1 - a_1 + z' + z + 2)$, and note that $\absolute{SP} \le q(x + 1)$ implies $x' \le {(q(x + 1) - q - r + s - q') / q} \le x$.
    Finally, the representation stated in the lemma generates the string $XQ_1Q_1[1\dd r-s]\rot^{r - s}(Q_1)^{x'}\rev{\hatQ_2}^{a_2}\rev{\hatQ_3}^{a_3}\dots\rev{\hatQ_t}^{a_t}$ of length $\absolute{X} + q + r - s + qx' + q' = \absolute{XSP}$. Hence $XSP \in \Aa'$ as required.
\end{proof}

\begin{corollary}\label{cor:inner}
Let $\angles{X, (Q_i, 1, u_i)_{i=1}^t}$ be a canonical representation of an affine prefix set $\Aa$. Let $\alpha \in \mathbb Q$ be the largest possibly fractional exponent such that $X Q_1^\alpha$ is a prefix of $T$, and define 
$$\Ss = \{S\cdot P : S \cdot P \text{ is a prefix of } X Q_1^\alpha, S \in \Aa, P \in \pal, \absolute{P} \ge 2 \absolute{Q_1}\}$$
There are $t' \le t$ affine prefix sets $\Bb_i$, $i \in [1, t']$, each of order $\le t$, such that both of the following properties hold for $\Bb = \bigcup_{i = 1}^{t'} \Bb_i$:
\begin{enumerate}
\item $\Ss \subseteq \Bb$.
\item For every $Y' \in \Bb$, there is a string $Y \in \Aa$ and a palindrome $P$ such that $Y' = YP$.
\end{enumerate}
\end{corollary}
\begin{proof}
If $\Ss$ is empty, then $t' = 0$ trivially satisfies the claim of the lemma.
Otherwise, note that any palindrome $P$ considered by $\Ss$ is a substring of $Q_1^{\alpha}$. Hence, if $\Ss$ is non-empty, there is a palindromic substring of $Q_1^{\alpha}$ that is of length at least $2\absolute{Q_1}$, and \cref{lem:aux:structure:palindrome_in_run:extended} implies that $\rev{Q_1}$ is a rotation of $Q_1$. Particularly, for arbitrary integer $s$, there is an integer $r \in [s, s + \absolute{Q_1})$ such that $\rev{Q_1} = \rot^{r}(Q_1)$.
This allows us to apply \cref{th:inner} to $X Q_1^{\floor{\alpha} + 3}$ to obtain an affine prefix set~$\Aa'$ of order $t$ satisfying each of the following:
    \begin{enumerate}
        \item If $Y' \in \Aa'$, then there is a string $Y \in \Aa$ and a palindrome~$P$ such that $Y' = YP$.
        \item For $Y \in \Aa$ and $P \in \pal$, if $\absolute{P} \ge 2\absolute{Q_1}$ and $YP$ is a prefix of $XQ_1^{\floor{\alpha} + 1}$, then $YP \in \Aa'$.
    \end{enumerate}
Let $\Bb = \Aa' |_{\absolute{X} + \alpha \cdot \absolute{Q_1}}$. 
We have $\Ss \subseteq \Bb$ and for all $Y' \in \Bb$, there is a string $Y \in \Aa$ and a palindrome $P$ such that $Y' = YP$. By  \cref{lem:truncate_affine_set}, $\Bb$ is a union of $\le t$ affine prefix sets of order $\le t$.
\end{proof}

\subsubsection%
[Appending a long palindrome outside a run of Q1]%
{Appending a long palindrome outside a run of \boldmath$Q_1$\unboldmath}

\begin{theorem}\label{thm:overhanging}
Let $\angles{X, (Q_i, 1, u_i)_{i=1}^t}$ be a canonical representation of an affine prefix set $\Aa$ and $s = \sum_{i = 2}^t (u_i + 1) \cdot \absolute{Q_i}$. For $j \in [1, t]$, let $\hatQ_j$ be the length-$\absolute{Q_j}$ suffix of $\rot^s(Q_1)$. For any string~$P$,%
$$\angles{X \cdot Q_1^{u_1+2} \cdot P \cdot \rev {Q_1}[1 \dd \absolute{Q_1}-s], (\rev{\hatQ_i}, 1, u_i)_{i = 1}^t}$$
represents an affine prefix set $\Aa'$ of the string $X \cdot Q_1^{u_1+2} \cdot P \cdot \rev{Q_1^{u_1+2}}$, where
$$\Aa' = \{ SWP \cdot \rev{W} \mid S \in \Aa\text{ and }SW = X \cdot Q_1^{u_1+2}\}.$$
\end{theorem}
\begin{proof}
Let $q = \absolute{Q_1}$. We can split the output representation into a concatenation
\begin{equation}\label{eq:overhanging}
\angles{X \cdot Q_1^{u_1+2} \cdot P \cdot \rev{Q_1}[1 \dd q-s], (\rev{\hatQ_1}, 1, u_1)} \cdot \angles{\emptystring, (\rev{\hatQ_i}, 1, u_i)_{i = 2}^t}.
\end{equation}

By \cref{lem:aux:structure:reverse_aps}, $\angles{\emptystring, (\rev{\hatQ_i}, 1, u_i)_{i = 2}^t}$ represents an affine prefix set of $\rev{\hatQ_1}$.
Hence any string generated by \cref{eq:overhanging} is a prefix of $X{Q_1^{u_1+2}}P \cdot {\rev {Q_1}[1 \dd q-s]} \cdot {\rev{\hatQ_1}^{u_1 + 1}}$. Due to $\rev{\hatQ_1} = \rev{\rot^{s}(Q_1)} = \rev{Q_1}[q - s + 1\dd q] \cdot \rev{Q_1}[1\dd q - s]$, this string is in turn a prefix of $XQ_1^{u_1+2}P \cdot \rev{Q_1}^{u_1 + 2}$. We have shown that \cref{eq:overhanging} represents an affine prefix set of the string $X \cdot Q_1^{u_1+2} \cdot P \cdot \rev{Q_1^{u_1+2}}$.

Every element in $\Aa$ contributes exactly one element to $\Aa'$, and hence $\absolute{\Aa'} = \absolute{\Aa}$. Whether or not an affine representation is irreducible depends solely on the lengths and exponent bounds of its components. Since lengths and exponent bounds are identical for the two representations stated in the lemma, it is clear that both of them are irreducible. By \cref{lem:aux:irreduciblecardinality}, each of the two representations generates exactly $\absolute{\Aa} = \absolute{\Aa'} = \prod_{i = 1}^{t} u_i$ distinct strings.

Since \cref{eq:overhanging} generates exactly $\absolute{\Aa'}$ distinct strings, it suffices to show that any string generated by \cref{eq:overhanging} is in $\Aa'$. It then readily follows that \cref{eq:overhanging} generates exactly $\Aa'$. 
Thus, consider any string $S'$ generated by \cref{eq:overhanging}. 
Such a string must be of the form $S' = XQ_1^{u_1 + 2}P \cdot \rev{W}$, where $\rev{W} = \rev{Q_1}[1\dd q-s] \cdot \rev{\hatQ_1}^{a_1}\rev{\hatQ_2}^{a_2} \dots \rev{\hatQ_t}^{a_t}$ for some exponents $\forall i \in [1, t]: a_1 \in [1, u_i]$. 
By our previous observations, $\rev{W}$ is a prefix of $\rev{Q_1}^{u_1 + 2}$, and thus there is a unique string $S$ such that $SW = XQ_1^{u_1 + 2}$ and $S' = SWP\cdot \rev{W}$.
It remains to be shown that $S \in \Aa$, which then implies $S' \in \Aa'$. For this purpose, we carefully analyze the length of $S$.
\begin{alignat*}{1}
\absolute{S} \enskip = \enskip \absolute{XQ_1^{u_1 + 2}} - \absolute{W} \enskip = \enskip & \absolute{XQ_1^{u_1 + 2}} - (q-s) - \sum\nolimits_{i = 1}^t a_i \cdot \absolute{Q_i}\\
= \enskip &\absolute{XQ_1^{u_1 - a_1  + 1}} + s - \sum\nolimits_{i = 2}^t a_i \cdot \absolute{Q_i}\\
= \enskip &\absolute{X} + \sum\nolimits_{i = 1}^t (u_i - a_i + 1) \cdot \absolute{Q_i}
\end{alignat*}
Recall that $\forall i \in [1, t] : a_i \in [1, u_i]$ and thus $(u_i - a_i + 1) \in [1, u_1]$. This means that $S = XQ_1^{u_1 - a_1 + 1}Q_2^{u_2 - a_2 + 1}\dots Q_t^{u_t - a_t + 1}$ is indeed in $\Aa$, which concludes the proof.
\end{proof}

If a fragment $P = T[x\dd y]$ of $T$ is a palindrome, denote its center $(x+y)/2$ by $\cent(P)$. 

\begin{corollary}\label{cor:overhanging}
Let $\angles{X, (Q_i, 1, u_i)_{i=1}^t}$ be a canonical representation of an affine prefix set $\Aa$, and consider the set of strings 
$$\Aa' = \{S\cdot P : S \cdot P \text{ is a prefix of } T, S \in \Aa, P \in \pal, \cent(P) > \absolute{XQ_1^{u_1 + 3}} \}.$$
T%
here are $t' = \Oh(t \log n)$ affine prefix sets $\Bb_i$, $i \in [1, t']$, each of order $\le t + 1$, such that both of the following properties hold for $\Bb = \bigcup_{i = 1}^{t'} \Bb_i$:
\begin{enumerate}
\item $\Aa' \subseteq \Bb$.
\item For every $Y' \in \Bb$, there is a string $Y \in \Aa$ and a palindrome $P$ such that $Y' = YP$.
\end{enumerate}
\end{corollary}
\begin{proof}
Consider any $S P \in \Aa'$, where $S P \text{ is a prefix of } T$, $S \in \Aa$, $P \in \pal$, $\cent(P) > \absolute{X Q_1^{u_1 + 3}}$. Due to $S \in \Aa$, \cref{lem:aux:structure:stronglength} implies $\absolute{S} < \absolute{X Q_1^{u_1 + 1}}$. 
Let $P' = T[x \dd y]$, where $x = 1 + \absolute{XQ_1^{u_1 + 2}}$ and $y = 2 \cdot  \cent(P) - x$. We claim that $P' \in \pal$. Indeed, the starting position $\absolute{S}+1$ of $P$ is less than the starting position $x$ of $P'$, and the centers of $P$ and $P'$ coincide with $\cent(P)-x = y-\cent(P)$.
We call $P'$ the \emph{core palindrome} of $SP$.
Note that every core palindrome is a prefix of $T[x\dd n]$ (which is independent of $SP$).
Therefore, by \cref{lm:base_case}, the set of core palindromes can be represented as the union of $\Oh(\log n)$ affine prefix sets. Let $\Cc$ be any of these sets.
We now describe how to compute the part of $\Aa'$ that contains strings of the form $SP = SWP' \cdot \rev{W}$, where $S \in \Aa$, $P \in \pal$, and the core palindrome of $SP$ is some $P' \in \Cc$. The procedure depends on the representation of $\Cc$, which, by \cref{lm:base_case}, is covered by one of the following cases. Let $q = \absolute{Q_1}$.

\def\Pi{P_{i}}
\def\Ci{\Cc_{i}}
\def\Wi{\Ww_{i}}
\def\WiTruncated{\Wi |_{(\alpha \cdot q)}}
\def\WwTruncated{\Ww |_{(\alpha \cdot q)}}
\def\WOneTruncated{\Ww_{1} |_{(\alpha \cdot q)}}
\def\WTruncated{\Ww'}
\subparagraph*{Case 1:} $\Cc$ is given in strongly affine representation $\angles{U \cdot (VU)^{\ell}, (VU, 1, u)}$, where $VU$ is primitive and $\absolute{VU} > q$. 
For any $i \in [1, u]$, let $\Pi = U \cdot (VU)^{\ell + i}$ be a core palindrome in $\Cc$. Using \cref{thm:overhanging}, we compute an affine prefix set $\Ci = \{ SWP_i \cdot \rev{W} \mid S \in \Aa\text{ and }SW = X \cdot Q_1^{u_1+2}\}$ of the string $XQ_1^{u_1 + 2}\Pi \cdot\rev{Q_1}^{u_1 + 2}$ (not necessarily a prefix of $T$) with a representation 
$$\angles{X \cdot Q_1^{u_1+2} \cdot P_i \cdot \rev {Q_1}[1 \dd \absolute{Q_1}-s], (\rev{\hatQ_i}, 1, u_i)_{i = 1}^t}.$$
If $\Aa'$ contains some string $SWP_i \cdot \rev{W} = XQ_1^{u_1+2}\Pi \cdot \rev{W}$, then this string is clearly also in~$\Ci$. 
Also, every string in $\Ci$ is the concatenation of some string in $\Aa$ and a palindrome. 
However, there may be some strings in $\Ci$ that are not prefixes of $T$ if $XQ_1^{u_1 + 2}\Pi \cdot \rev{Q_1}^{u_1 + 2}$ is not a prefix of~$T$. 
We have already established that $XQ_1^{u_1 + 2}\Pi$ is a prefix of $T$, and we split the representation of~$\Ci$ into two parts by defining a set $\Wi$ such that $\Ci = \{ XQ_1^{u_1 + 2}\Pi \cdot \rev{W} \mid \rev{W} \in \Wi \}$. 
The first part has representation $\angles{XQ_1^{u_1 + 2}\Pi, \emptyseries}$ of order 0. The second part $\Wi$ has representation
$$\angles{\rev {Q_1}[1 \dd \absolute{Q_1}-s], (\rev{\hatQ_i}, 1, u_i)_{i = 1}^t}$$
and contains prefixes of $\rev{Q_1}^{u_1 + 2}$. Let $x = 1 + \absolute{XQ_1^{u_1 + 2}}$ and $y_i = x + \absolute{P_i} - 1$, i.e., $T[x\dd y_i] = P_i$. Our goal is to truncate $\Wi$ such that we remove exactly all the strings that are not prefixes of $T[y_i + 1\dd n]$. Note that all the $\Wi$'s are identical, i.e., they are independent of $i$, and we will show that they remain identical even after truncating.

Recall that the given representation of $\Cc$ is strongly affine, and hence, $T[x\dd n]$ has a prefix $U \cdot (VU)^{\ell + u + \strongconst} = P_i \cdot (VU)^{\strongconst + u - i}$ (for any $i\in[1,u]$). This implies that $T[y_i + 1\dd n]$ has a prefix $(VU)^2$. 
Let $\alpha \in \mathbb Q$ be the largest (possibly fractional exponent) such that $\rev{Q_1}^\alpha$ is a prefix of $(VU)^2$. Since $(VU)$ is primitive, and due to $q < \absolute{VU}$, it cannot be that $q$ is a period of $(VU)^2$ (see \cref{lem:primitive_squares}). 
Hence $\absolute{\rev{Q_1}^\alpha} < \absolute{(VU)^2}$, and $\alpha$ is the maximal exponent such that $\rev{Q_1}^\alpha$ is a prefix of $T[y_i + 1\dd n]$.
We apply \cref{lem:truncate_affine_set} and obtain $\WiTruncated$ as the union of $t' \le t$ representations, each of order at most $t$. Note that $\alpha$ is independent of $i$, and thus all the $\WiTruncated$ are indeed identical.

Finally, let $\WTruncated = \WOneTruncated$. We must represent the union of all the truncated $\Ci$'s, defined by $\Bb(\Cc) = \{ XQ_1^{u_1 + 2} \cdot P_i \cdot \rev{W} \mid i \in [1, u], \rev{W} \in \WTruncated\}$.
This set can be readily obtained by creating $t'$ copies of $\angles{XQ_1^{u_1 + 2}U(VU)^{\ell}, (VU, 1, u)}$, and concatenating each copy with one of the $t'$ representations of order at most $t$ that make up $\WTruncated$. 
Then, $\Bb(\Cc)$ is the union of $t$ affine prefix sets, each of order at most $t + 1$.

\subparagraph*{Case 2:} $\Cc$ is given in representation $\angles{P', \emptyseries}$ of order $0$, i.e., it contains a single core palindrome~$P'$. 
We proceed exactly like in Case 1, but with a single palindrome $P_1 = P'$. 
In Case 1, we only use the periodic structure of the $\Pi$'s to show that all the truncated sets $\WiTruncated$ are identical. 
Since this time we are only concerned with a single core palindrome, we simply let $\alpha \in \mathbb Q$ be the maximal value such that $\rev{Q_1}^\alpha$ is a prefix of $T[y_1 + 1\dd n]$. 
We then continue just like in Case 1, performing the final step with the fixed set $\angles{XQ_1^{u_1 + 2}P_1, \emptyseries}$ instead of $\angles{XQ_1^{u_1 + 2}U(VU)^{\ell}, (VU, 1, u)}$. 
This results in a set $\Bb(\Cc)$ that is the union of at most $t$ representations, each of order at most $t$.

\subparagraph*{Case 3:} $\Cc$ has strongly affine representation $\angles{U \cdot (VU)^{\ell}, (VU, 1, u)}$, where $VU$ is primitive and $\absolute{VU} = q$. 
For $i \in [1, u]$, let $\Pi = U \cdot (VU)^{\ell + i}$. 
We show that, if $\Aa'$ contains some $SP = S \cdot W \cdot P_i \cdot \rev{W} = XQ_1^{u_1+2}\Pi \cdot \rev{W}$ with $S \in \Aa$, then the entire $SP$ can be written as $XQ_1^{\alpha}$ for some $\alpha \in \mathbb Q$.
By the definition of a core palindrome, the center of $P_i$ is $\cent(P_i) > \absolute{XQ_1^{u_1 + 3}}$, and its length is $2 \cdot (\cent(P_i) - x) + 1 \ge 2q$.
The entire palindrome $P = WP_i \cdot \rev{W}$ is then also of length at least $2q$. 
It holds $P_i[1\dd 2q] = Q_1^2$ because $T$ has a prefix $XQ_1^{u_1 + \strongconst}$ (since $\Aa$ is given in strongly affine representation) and $P_i$ starts at position $x = 1 + \absolute{XQ_1^{u_1 + 2}}$. 
It follows $SWP_i[1\dd 2q] = XQ_1^{u_1 + 4}$ with $\absolute{S} > \absolute{X}$. %Since $P_i$ has period $q$, it follows that 
Therefore, $P$ has the $q$-periodic prefix $WP_i$ of length $(\absolute{P} + \absolute{P_i})/2 \ge \absolute{P}/2 + q > 3q/2$, and \cref{lem:aux:structure:palperiodoverlap} implies that $P$ has period $q$. We have established that $SWP_i[1\dd q] = XQ_1^{u_1 + 3}$, and that $P_i \cdot \rev{W}$ has period $q$. 
Since these fragments overlap by $q$ symbols, it is clear that $SWP_i \cdot \rev{W}$ is of the form $XQ_1^\alpha$ for some exponent $\alpha \in \mathbb Q$.

We have shown that Case 3 is only concerned with prefixes of the form $XQ_1^{\alpha}$ and with palindromes of length at least $2q$. 
Hence we can simply apply \cref{cor:inner} and obtain a set $\Bb(\Cc)$ as the union of at most $t$ representations of order at most $t$. 
This set then contains every $SP = SWP_i \cdot \rev{W} = XQ_1^{u_1+2}\Pi \cdot \rev{W} \in \Aa'$ for all the core palindrome sets $\Cc$ that fall into Case 3, and \cref{cor:inner} guarantees that any element in $\Bb(\Cc)$ is indeed the concatenation of a string in~$\Aa$ and a palindrome. 

\subparagraph*{Case 4:} $\Cc$ has strongly affine representation $\angles{U \cdot (VU)^{\ell}, (VU, 1, u)}$, where $VU$ is primitive and $\absolute{VU} < q$. 
As before, let $x = 1 + \absolute{XQ_1^{u_1 + 2}}$, and let $P_i = U \cdot (VU)^{\ell + i}$. 
As seen in Case~3, $P_i$ is of length at least $2q$ and has prefix $P_i[1\dd 2q] = Q_1^2$. Since $P_i$ has a period $\absolute{VU} < q$, its prefix $Q_1^2$ also has a period $\absolute{VU}$. However, this contradicts the fact that $Q_1$ is primitive (see \cref{lem:primitive_squares}).

\vspace{.5\baselineskip}
The list of cases is clearly exhaustive, and it remains to analyze the number of created representations. There are $\Oh(\log n)$ core palindrome sets, and each set is covered by exactly one case. In every case, we create at most $t$ representations, each of order at most $t + 1$, matching the statement of the corollary.
\end{proof}

\subsubsection{Appending all long palindromes}

\begin{lemma}\label{lem:new_palindrome_irreducible_sets}
Let $\angles{X, (Q_i, 1, u_i)_{i=1}^t}$ be a canonical representation of an affine prefix set $\Aa$. Define the set $\Aa' = \{S \cdot P \mid S \in \Aa, P \in \pal, |P| \ge 2\absolute{Q_1}\text{, and $S\cdot P$ is a prefix of $T$} \}$. 
There are $t' = \Oh(t \log n)$ affine prefix sets $\Bb_i$, $1 \le i \le t'$, each of order at most $t + 1$, that satisfy both of the following:
\begin{enumerate}
\item $\Aa' \subseteq \cup_{i = 1}^{t'} \Bb_i$.
\item For each string $S' \in \cup_{i = 1}^{t'} \Bb_i$, there is a string $S \in \Aa$ and $P \in \pal$ such that $S' = S \cdot P$.
\end{enumerate}
\end{lemma}
\begin{proof}
We consider the sets from \cref{cor:inner,cor:overhanging}, defined by%
$$\Aa_1= \{S\cdot P : S \cdot P \text{ is a prefix of } X Q_1^\alpha, S \in \Aa, P \in \pal, \absolute{P} \ge 2 \absolute{Q_1}\}\text{ and}$$
$$\Aa_2 = \{S\cdot P : S \cdot P \text{ is a prefix of } T, S \in \Aa, P \in \pal, \cent(P) > \absolute{X} + (u_1+3) \cdot \absolute{Q_1}\},$$
where $\alpha$ is the largest (possibly fractional) exponent such that $XQ_1^\alpha$ is a prefix of $T$. Due to \cref{cor:inner,cor:overhanging}, we can express (a superset of) $\Aa_1 \cup \Aa_2$ as the union of at most $\Oh(t \log n)$ affine prefix sets, each of order at most $t + 1$, where every string in each of the prefix sets is the concatenation of a string from $\Aa$ and a palindrome. 

It remains to be shown that $\Aa' \subseteq \Aa_1 \cup \Aa_2$.
For the sake of contradiction, assume that there is some string $SP \in \Aa' \setminus (\Aa_1 \cup \Aa_2)$, where $S \in \Aa$, $P \in \pal$ and $\absolute{P} \ge 2 \absolute{Q_1}$.
Due to $SP \notin \Aa_1$, it cannot be that $SP$ is a prefix of $XQ_1^{\alpha}$. Thus, $SP$ must be longer than $XQ_1^{\alpha}$. Let $m = \absolute{XQ_1^{\alpha}} - \absolute{S}$, then the length-$m$ suffix of $Q_1^{\alpha}$ is a prefix of $P$. We show a lower bound on $m$.
Since the given representation is strongly affine, it holds $\alpha \ge u_1 + 5$. It is also irreducible, and hence \cref{lem:aux:structure:stronglength} implies $\absolute{S} < \absolute{XQ_1^{u_1 + 1}}$.  Therefore, it holds $m > 4\absolute{Q_1}$.
Note that $P$ does not have period $\absolute{Q_1}$, but its length-$m$ prefix does. Hence, by \cref{lem:aux:structure:palperiodoverlap}, it follows that $P$ is of length over $2m - \absolute{Q_1}$, and therefore
$$\cent(P) \ge \absolute{S} + \absolute{P}/2 > \absolute{S} + m - \absolute{Q_1}/2 = \absolute{XQ_1^{\alpha}} - \absolute{Q_1}/2 > \absolute{XQ_1^{u_1 + 4}}.$$
This implies $SP \in \Aa_2$, which contradicts the initial assumption.
\end{proof}

\subsection{Recursively appending shorter palindromes and the final result}

We have shown that appending palindromes of length at least $2\absolute{Q_1}$ results in at most $\Oh(t \log n)$ affine prefix sets of order at most $t + 1$. For appending shorter palindromes, we will exploit properties of strongly affine prefix sets that allow us to apply the previously described approach recursively.

\begin{lemma}\label{lem:new_palindrome_recursive}
Let $\angles{X, (Q_i, 1, u_i)_{i=1}^t}$ be a canonical representation of an affine prefix set $\Aa$. Define the set of strings 
$$\Aa' = \{S\cdot P : S \cdot P \text{ is a prefix of } T, S \in \Aa, P \in \pal\}.$$
There are $t' = \Oh((t + 1)^2 \log n)$ affine prefix sets $\Bb_i$, $1 \le i \le t'$, each of order at most $t + 1$, such that $\Aa' = \cup_{i = 1}^{t'} \Bb_i$.
\end{lemma}

\begin{proof}
    For the sake of induction, consider the case where $t = 0$. It holds $\Aa = \{ X \}$ and $\Aa' = \{ XP \mid P \in \pal\text{ and $P$ is a prefix of $T[1+\absolute{X}\dd n]$} \}$. By \cref{lm:base_case}, we can express the prefix-palindromes of $T[1+\absolute{X}\dd n]$ as the union of $\Oh(\log n)$ affine prefix sets, and by prepending $X$ we immediately obtain the statement of the lemma.

    Now we show that the lemma holds for representations of order $t > 0$, inductively assuming that we have already shown the correctness for representations of order $t - 1$.
    We apply \cref{lem:new_palindrome_irreducible_sets} and obtain $\Oh(t \log n)$ affine prefix sets, each of order at most $t + 1$. These sets correspond to a superset of the prefixes in
    $$\Aa'_{\text{long}} = \{S\cdot P : S \cdot P \text{ is a prefix of } T, S \in \Aa, P \in \pal, \absolute{P} \ge 2 \absolute{Q_1}\},$$
    and each prefix contained in any of the sets is the concatenation of an element in $\Aa$ with a palindrome.
    It remains to be shown how to cover
    $$\Aa'_{\text{short}} = \{S\cdot P : S \cdot P \text{ is a prefix of } T, S \in \Aa, P \in \pal, \absolute{P} < 2 \absolute{Q_1}\}.$$
    
    By \cref{lem:aux:structure:aps_of_q1:strong}, $\angles{\emptystring, (Q_i, 1, u_i)_{i=2}^t}$ is a canonical representation of an affine prefix set $\Aa''$ of~$Q_1^2$.
    By the inductive assumption, the set
    $$\Aa''' = \{S\cdot P : S \cdot P \text{ is a prefix of } Q_1^3, S \in \Aa'', P \in \pal\}$$
    is the union of $t'' = \Oh(t^2 \log n)$ %\jonas{this could be $\log \absolute{Q_1}$ rather than $\log n$, probably does not matter}
    affine prefix sets $\Bb_j$, $j \in [1, t'']$ of order at most $t$.
    For any set $\Bb_j$, let $\angles{X', (Q'_i, \ell_i', u_i')_{i = 1}^{t_j}}$ be a representation of order $t_j \le t$.
    We obtain the set $\Bb_j'$ defined by the representation $\angles{X, (Q_1, 1, u_1) \cdot (Q'_i, \ell_i', u_i')_{i = 1}^{t_j}}$ of order $t_j + 1 \le t + 1$.
    To conclude the proof we show the following claims:
    \begin{enumerate}
        \item Each set $\Bb_j'$ is an affine prefix set of $T$.\label{claimisAPS}
        \item Every element in any of the sets $\Bb_j'$ is also in $\Aa'$.\label{claimissubset}
        \item Every element in $\Aa_{\text{short}}'$ is contained in at least one set $\Bb_j'$.\label{claimissuperset}    
    \end{enumerate}
    For Claim~\ref{claimisAPS}, observe that a string in $\Bb_j$ is a prefix of $Q_1^3$, and thus a string in $\Bb_j'$ is a prefix of $XQ_1^{u_1 + 3}$. Since the original representation is strongly affine, $XQ_1^{u_1 + 3}$ is a prefix of $T$, and the correctness of the claim follows.

    For Claim~\ref{claimissubset}, every element in $\Bb_j'$ is of the form $XQ_1^{a_1}SP$ for a palindrome $P$, some $S \in \Aa''$, and $a_1 \in [1, u_1]$. Note that elements in $\Aa''$, hence also $S$, are of the form $Q_2^{a_2}Q_3^{a_3}\dots Q_t^{a_t}$ with $\forall i \in [2, t] : a_i \in [1, u_i]$. Thus, $XQ_1^{a_1}S$ is in $\Aa$, and $XQ_1^{a_1}SP$ is in $\Aa'$ as claimed.

    For Claim~\ref{claimissuperset}, consider an element $SP \in \Aa'_{\text{short}}$ where $S \in \Aa$ and $P \in \pal$ with $\absolute{P} < 2 \absolute{Q_1}.$
    We can write $S$ as $XQ_1^{a_1}Q_2^{a_2}\dots Q_t^{a_t}$ with $\forall i \in [1, t] : a_i \in [1, u_i]$. 
    By \cref{lem:aux:structure:stronglength}, $Q_2^{a_2}\dots Q_t^{a_t}$ is of length less than $Q_1$, and $SP$ is of length less than $\absolute{XQ_1^{a_1 + 3}} \le \absolute{XQ_1^{u_1 + 3}}$. Since the original representation is strongly affine, $XQ_1^{u_1 + 3}$ is a prefix of $T$. 
    Since also $XQ_1^{a_1}Q_2^{a_2}\dots Q_t^{a_t}P$ is a prefix of $T$, it is easy to see that $Q_2^{a_2}\dots Q_t^{a_t}P$ is a prefix of $Q_1^3$. Finally, $Q_2^{a_2}\dots Q_t^{a_t}$ is in $\Aa''$, and thus $Q_2^{a_2}\dots Q_t^{a_t}P$ is in $\Aa'''$. It follows that $SP = XQ_1^{a_1}Q_2^{a_2}\dots Q_t^{a_t}P$ is in one of the $\Bb_j'$.

    We created $\Oh(t \log n)$ representations of order at most $t + 1$ that cover $\Aa_{\text{long}}'$, and $t'' = \Oh(t^2 \log n)$ representations of order at most $t + 1$ that cover $\Aa_{\text{short}}'$. Hence we created $\Oh((t + 1)^2 \log n)$ representations of order at most $t + 1$, as required by the lemma.
\end{proof}

\structuralthm*

\begin{proof}
    We start with the empty affine prefix set representing $\pal^0$. We proceed in $k$ levels $k' \in  [0, k)$, and, on each level $k'$, we consider affine prefix sets of order $k'$. 
    The union of all the affine prefix sets of level $k'$ are exactly all of the prefixes of $T$ that are in $\pal^{k'}$. 
    For each affine prefix set of the current level $k'$, we first apply \cref{lem:aux:makeirreducible} and \cref{lem:aux:makestrong} to obtain $6^{k'}$ canonical representations of order at most $k'$. 
    Then, for each of the representations, we append a palindrome using \cref{lem:new_palindrome_recursive}, resulting in $c \cdot (k' + 1)^2 \log n$ affine prefix sets of order at most $k' + 1$, which we move to level $k' + 1$. 
    Here, $c$ is a positive constant that depends on the precise complexity analysis of \cref{lem:new_palindrome_recursive}. 
    Hence, after processing level $k - 1$, the total number of affine prefix sets is bounded by $$\prod_{k' = 0}^{k - 1} (6^{k'} \cdot c \cdot (k' + 1)^2 \log n) \le (k!)^2 \cdot c^k \cdot 6^{(k^2/2)} \cdot \log^k n.$$
    Let $\smallconst \in \mathbb R^+$ with $\smallconst < 1$ be constant. If $k$ exceeds another constant that depends solely on $\smallconst$ and $c$, then $(k!)^2 \cdot c^k < 6^{\smallconst' \cdot k^2}$ with $\smallconst' = \frac{1}{2-\smallconst} - \frac12 > 0$. Hence the bound becomes $6^{\smallconst' \cdot k^2} \cdot 6^{(k^2/2)} \cdot \log^k n = 6^{k^2/(2-\smallconst)}$.
\end{proof}

\section%
[Read-only Algorithm for Recognising k-palindromic Prefixes and Computing the Palindromic Length]%
{Read-only Algorithm for Recognising \boldmath$k$\unboldmath-palindromic Prefixes and Computing the Palindromic Length}
\label{sec:algo}
Finally, using the combinatorial properties we showed above, we develop a read-only algorithm that receives a string $T$ of length $n$ and an integer $k$, and computes a small space-representation of $i$-palindromic prefixes of a string for all $1 \le i \le k$. Additionally, it can computes the palindromic length of $T$ if its at most $k$.

The following lemma gives a small-space implementation of the procedure behind \cref{coro:affine-set-pal1} to compute the prefix-palindromes of a string as $\Oh(\log n)$ affine prefix sets of order at most 1. The general idea is to enumerate the prefix-palindromes in the increasing order of length, using constant-space exact pattern matching.

\begin{lemma}\label{lem:enum_prefixpali}
There is a read-only algorithm that enumerates all prefix-palindromes of a string $T[1 \dd n]$ in the increasing order of length in $\Oh(n)$ time and $\Oh(1)$ space.
\end{lemma}

\begin{proof}
For a fixed $j \in [0, \ceil{\log_2 n}]$, we show how to enumerate the prefix-palindromes of length within range $[2^{j}, 2^{j + 1})$. Consider any length $m \in [2^{j}, 2^{j + 1})$. It is easy to see that $T[1\dd m]$ is a palindrome if and only if $\rev{T[1\dd 2^j]} = T(m - 2^j \dd m]$ because $T[1 \dd 2^j]$ spans more than half of $T[1 \dd m]$. Recall that for strings $X,Y$, a fragment $Y[i \dd j] = X$ is an occurrence of $X$ in $Y$. 
We use constant space and linear time pattern matching (see~\cite{DBLP:journals/tcs/BreslauerGM13} and references therein) to enumerate all occurrences of a string $\rev{T[1 \dd 2^j]}$ in $T[1 \dd 2^{j + 1})$ in the left-to-right order.
Whenever the pattern matching algorithm outputs an occurrence $T[i \dd i+2^j) = \rev{T[1\dd 2^j]}$, we output that $T[1\dd m]$ with $m = i + 2^j - 1$ as a palindromic prefix.
By the previous observation, this reports all palindromic prefixes of length within range $[2^{j}, 2^{j + 1})$ in the increasing order of length, using constant space and $\Oh(2^{j})$ time. Hence the total time for all $j \in [0, \ceil{\log_2 n}]$ is $\Oh(\sum_{j = 1}^{\ceil{\log_2 n}} 2^j) = \Oh(n)$.
\end{proof}

\begin{restatable}[{Implementation of \cref{lm:base_case}}]{algo}{roalgopalone}\label{lem:ro-algo-pal1}
    Given a string $T[1\dd n]$, there is a read-only algorithm that uses $\Oh(\log n)$ space and $\Oh(n)$ time and outputs all prefixes of $T$ that belong to $\pal$ as $\Oh(\log n)$ affine sets of order at most 1. Each set of order $1$ is reported in canonical representation $\angles{U(VU)^\ell, (VU, 1, u)}$ for some $U\in \pal\cup \{\emptystring\}$, $V\in \pal$ and integers $\ell \ge 1$ and $u > 1$.
\end{restatable}
\begin{proof}
    We use the procedure described in the proof of \cref{coro:affine-set-pal1}.
    We start with a single prefix-palindrome represented by $\angles{T[1], \emptyseries}$.
    Then, we use \cref{lem:enum_prefixpali} to enumerate the remaining prefix-palindromes in increasing order of length.
    Let $P'$ be the most recently reported prefix-palindrome (initially $P' = T[1]$).
    Whenever some prefix $P = T[1\dd m]$ is reported to be a palindrome, we consider two cases based on whether or not $\absolute{P} > 3\absolute{P'}/2$. We proceed as in the proof of \cref{coro:affine-set-pal1}, storing all affine sets computed during the procedure (where creating or modifying a set takes constant time, performing simple arithmetic operations that depend solely on $\absolute{P}$ and $\absolute{P'}$). 
    The representations may not be strongly affine, but the postprocessing described in \cref{lem:aux:makestrong} (and used to obtain \cref{lm:base_case}) can be easily performed in constant time for each affine prefix set. (This holds because the present representations are of order 1.)
    There are at most $\Oh(\log n)$ affine sets, hence the algorithm uses $\Oh(\log n)$ space and $\Oh(n)$ time.
\end{proof}

\roalgopalk*
\begin{proof}
We first provide algorithmic implementations of the combinatorial lemmas we showed above (rather straightforward, but we still provide them for completeness), and then combine them into the final algorithm. 

We assume that a representation $\angles{X, (Q_i, \ell_i, u_i)}$ of an affine set is stored as a list of the components $Q_i$, associated with $\ell_i, u_i$, and $X$ is stored separately. 
\subsubsection*{Transforming representations}
\begin{algo}[Implementation of \cref{lem:aux:transform}]\label{algo:transform}
    Given a pointer to the $i$-th component of the representation of an affine set,
    the four operations $\textop{switch}_i$, $\textop{merge}_i$, $\textop{split}_i$, and $\textop{truncate}$ can each be performed in constant time and space.
\end{algo}
These four operations only change $\ell_i, u_i$, or replace $Q_i$ with its rotation, which can be computed via a constant number of arithmetic operations. In case of $\textop{merge}_i$, the $(i + 1)$-th element in the list of components needs to be deleted, while $\textop{split}_i$ requires an insertion between the $i$-th and $(i + 1)$-th element. For $\textop{switch}_i$, we need to swap the $i$-th and $(i + 1)$-th element. Given a pointer to the $i$-th element, the required deletions, insertions, and swaps take constant time. Hence all four operations can be implemented in $\Oh(1)$ time and $\Oh(1)$ extra space.
As a corollary, we obtain the following:
\begin{algo}[Implementation of \cref{lem:aux:removefixed}]\label{algo:removefixed}
    Given the representation of an affine prefix set of order $t$, we can remove all fixed components in $\Oh(t^2)$ time and $\Oh(1)$ space.
\end{algo}

\begin{restatable}[{Implementation of \cref{lem:aux:makeirreducible}}]{algo}{algoirreducible}\label{algo:make-irreducible}
    Let $\Aa$ be an affine prefix set of $T[1\dd n]$, given in representation $\rho$ of order $t$. There is a read-only algorithm that, given $\rho$ and random access to $T$, transforms $\rho$ into an irreducible representation $\rho^*$ of $\Aa$ in $\Oh(t^2)$ time and $\Oh(1)$ additional space. The representation~$\rho^*$ is of order at most $t$.
\end{restatable}

We assume to receive an affine prefix set with a representation of order $t$, and transform it into an irreducible representation. The transformation starts with an application of \cref{lem:aux:removefixed}, which requires $\Oh(t^2)$ time and $\Oh(1)$ extra space. It then performs $\le t$ $\textop{merge}$ operations and $\le t$ $\textop{split}$ operations in $\Oh(t^2)$ total time and $\Oh(1)$ extra space. Finally, the transformation applies \cref{lem:aux:removefixed} again in $\Oh(t^2)$ time and $\Oh(1)$ extra space. 

\subsubsection*{Strongly affine representations}

\begin{restatable}[Implementation of \cref{lem:aux:makestrong}]{algo}{makestrongalgo}\label{algo:makestrong}
    Given an affine prefix set $\Aa$ as a representation of order $t$, there is a read-only algorithm that computes, in $\Oh(\strongconstplus^t \cdot t^2)$ total time and $\Oh(\strongconstplus^t \cdot t)$ space,
    at most $6^t$ canonical representations of order at most $t$ of affine sets that form a partition of $\Aa$.
\end{restatable}

We start by applying \cref{algo:make-irreducible} to make the representation irreducible in $\Oh(t^2)$ time and $\Oh(1)$ space. We then create $6^t$ representations of \cref{eq:strong} naively in $\Oh(6^t)$ time and $\Oh(6^t \cdot t)$ space, and finally once more apply \cref{algo:make-irreducible} to each representation in total $\Oh(6^t \cdot t^2)$ time and $\Oh(6^t \cdot t)$ space. As observed in \cref{sec:strongaffine}, this indeed results in canonical representations.

\subsubsection*{Appending a new palindrome}
\begin{algo}[Implementation of \cref{lem:truncate_affine_set}]\label{algo:truncate_affine_set}
    Given an affine prefix set $\Aa$ with its representation $\angles{X, (Q_i, \ell_i, u_i)_{i = 1}^t}$ and an integer $m \in \mathbb{N}$, there is a read-only algorithm that finds a decomposition of $\Aa|_m$ into at most $t$ affine prefix sets, each of order at most $t$, in $\Oh(t^2)$ time and space.
\end{algo}
The representations of the sets are computed in a loop over $i = 1, 2, \ldots, t$.
For a fixed $i$, the task is, given exponents $a_1, a_2, \ldots, a_{i-1}$, to compute the minimal $a_i \in [\ell_i, u_i]$ such that $$\absolute{XQ_1^{a_1} Q_2^{a_2} \ldots Q_{i-1}^{a_{i-1}} Q_i^{a_i} Q_{i+1}^{u_{i+1}} \ldots Q_{t}^{u_{t}}} \ge m.$$
We then add a set with a representation $\angles{XQ_1^{a_1}Q_2^{a_2}\dots Q_{i - 1}^{a_{i - 1}}, (Q_i, \ell_i,a_i-1) \cdot (Q_j, \ell_j, u_j)_{j = i+1}^t}$ to the union, but only if $a_i > \ell_i$. Either way, we continue with $i+1$.
We can obtain $a_1$ by first computing $m' = \absolute{XQ_1^{a_1} Q_2^{a_2} \ldots Q_{i-1}^{a_{i-1}}Q_{i+1}^{u_{i+1}} \ldots Q_{t}^{u_{t}}}$, which clearly takes $\Oh(t)$ time. Then, it holds $a_i = \ceil{(m - m')/\absolute{Q_i}}$ (unless this value is not in $[\ell_i, u_i]$, which can be handled trivially). Creating the new representation takes $\Oh(t)$ time and space, and the overall time and space for all $i$ is $\Oh(t^2)$.

\begin{algo}[Implementation of \cref{cor:inner}]\label{algo:inner}
    Let $\angles{X, (Q_i, \ell_i, u_i)_{i = 1}^t}$ be a canonical representation of an affine prefix set $\Aa$. Let $\alpha$ be the largest (possibly fractional) exponent such that $X Q_1^\alpha$ is a prefix of $T$ and define 
    \[\Ss = \{S \cdot P: S \cdot P \text{ is a prefix of } X Q_1^\alpha, S \in \Aa, P \in \pal, \absolute{P} \ge 2 \absolute{Q_1}\}\]
There is a read-only algorithm that computes in $\Oh(n)$ time and $\Oh(t^2)$ space at most $t$ affine prefix sets $\Bb_i$ such that their union covers $\Ss$ and such that for every string $Y' \in \cup_i \Bb_i$ there is a string $Y \in \Aa$ and $P \in \pal$ satisfying $Y' = Y \cdot P$. Each set is reported in representation of order at most~$t$.
\end{algo}
We first compute $\alpha$, which takes $\Oh(n)$ time by naive scanning.
We then apply \cref{th:inner}, which gives an affine set $\Bb$ in representation of order $t$ in $\Oh(t)$ time and space. Finally, we truncate $\Bb$ to prefixes of length $\absolute{XQ_1^\alpha}$ in $\Oh(t^2)$ time and space with \cref{algo:truncate_affine_set}, resulting in at most $t$ representations of order at most $t$.

\begin{algo}[Implementation of \cref{cor:overhanging}]\label{algo:overhanging}
    Let $\angles{X, (Q_i, \ell_i, u_i)_{i = 1}^t}$ be a canonical representation of an affine prefix set $\Aa$. There is a read-only algorithm that computes a decomposition of a superset of 
    \[\Aa' = \{S \cdot P: S \cdot P \text{ is a prefix of } T, S \in \Aa, P \in \pal, \cent(P) \ge \absolute{X} + (u_1+3) \cdot \absolute{Q_1}\}\]
    (with the properties claimed in \cref{cor:overhanging}) into $\Oh(t \log n)$ affine prefix sets of order $\le t+1$ in $\Oh(n)$ time and $\Oh(t^2 \log n)$ space.
\end{algo}

We start by searching the prefixes of $T[\absolute{XQ_1^{u_1 + 2}} + 1 \dd]$ that belong to $\pal$. By \cref{lem:ro-algo-pal1}, we can find $\Oh(\log n)$ affine sets of order at most 1 containing all such prefixes in $\Oh(n)$ time and $\Oh(\log n)$ space. 
The sets of order $1$ are reported in strongly affine representation $\angles{U(VU)^\ell, (VU, 1, u)}$ for some $U\in \pal\cup \{\emptystring\}$, $V\in \pal$ and integers $\ell \ge 1$ and $u > 1$.

As seen in Case 4 of the proof of \cref{cor:overhanging}, it cannot be that $\absolute{VU} < \absolute{Q_1}$. Also, as shown in Case 3, if $\absolute{VU} = \absolute{Q_1}$, then we can simply apply \cref{algo:inner}, which takes $\Oh(n)$ time and $\Oh(t^2)$ space. Hence we only have to consider the cases where $\absolute{VU} > \absolute{Q_1}$ (Case 1), or where the affine prefix set is of order $0$ (Case 2). 
For Case 1, let $\rho_1 = \angles{XQ_1^{u_1 + 2}U(VU)^{\ell}, (VU, 1, u)}$, and let $P' = U(VU)^{\ell + u}$. 
We obtain $\rho_2 = \angles{\rev {Q_1}[1 \dd \absolute{Q_1}-s], (\rev{\hatQ_i}, 1, u_i)_{i = 1}^t}$ in $\Oh(t)$ time, where~$s$ and the $\hatQ_i$ are defined as in \cref{thm:overhanging}.
Then, we compute the maximal $\alpha \in \mathbb Q$ such that $XQ_1^{u_1 + 2}P' \cdot \rev{Q_1}^\alpha$ is a prefix of $T$, which can be done naively in $\Oh(n)$ time.
Now we truncate $\rho_2$ such that the generated strings are of length at most $\absolute{Q_1^\alpha}$, which takes $\Oh(t^2)$ time and space with \cref{algo:truncate_affine_set}, and results in at most $t$ representations of order at most $t$.
Finally, we concatenate a copy of $\rho_1$ with each of these representations.
In theory, when concatenating $\rho_1$ with some representation $\rho_2' = \angles{X', (Q'_i, \ell'_i, u'_i)}$, we have to compute the primitive root $Y$ of $X'$. 
This is because $(Y, \frac{\absolute{X'}}{\absolute{Y}}, \frac{\absolute{X'}}{\absolute{Y}})$ will become a component of the concatenation.
However, we can omit this step by observing that the removal of fixed components with \cref{algo:removefixed} does not depend on the primitiveness of fixed components, and hence we can create and immediately remove a non-primitive component $(X', 1, 1)$ instead.
Thus, we take $\Oh(t^2)$ time per concatenation, or $\Oh(t^3)$ time overall.%(which is $\Oh(\log^3 n) \subset \Oh(n)$ by \cref{lem:irreducible_maxorder}).
We followed the computational steps in Case 1 of the proof of \cref{cor:overhanging}, and the total time and space complexity (for one affine set of core palindromes) are respectively $\Oh(n + t^3)$ and $\Oh(t^2)$.

As seen in the proof of \cref{cor:overhanging}, if an affine set of prefix-palindromes of $T[\absolute{XQ_1^{u_1 + 2}} + 1 \dd]$ is given in representation $\angles{P', \emptyseries}$ (i.e., in Case 2), then we need a subset of the operations used for Case 1, leading to (at most) the same time and space complexity.

Summing over the $\Oh(\log n)$ sets of core palindromes, the total time and space complexity are respectively $\Oh(n \log n + t^3 \log n)$ and $\Oh(t^2 \log n)$. 
This can be improved by observing that processing a single set of core palindromes takes $\Oh(t^3)$ time if we ignore the $\Oh(n)$ time needed to find the maximal $\alpha \in \mathbb Q$ such that $XQ_1^{u_1 + 2}P' \cdot \rev{Q_1}^\alpha$ is a prefix of $T$.
We only use $\alpha$ to truncate $\rho_2$ to strings of length at most $\absolute{Q_1^\alpha}$, and every string generated by $\rho_2$ is of length over $\absolute{Q_1}$. Hence we only have to compute $\alpha$ if $\alpha > 1$.
In \cref{lem:corepali_extension} (below), we show how to perform this task in batch in overall $\Oh(n)$ time, using $\Oh(\log n)$ space. Hence the total time
%is $\Oh(n)$.
becomes $\Oh(n + t^3 \cdot \log n)$, which is $\Oh(n)$ due to \cref{lem:irreducible_maxorder}.

\begin{lemma}\label{lem:corepali_extension}
Let $T[1\dd n]$ be a string and let $Q$ be a fragment of $T$. Let $S_1, \dots S_h$ be prefixes of~$T$, in increasing order of length. There is a read-only algorithm that, in $\Oh(n)$ time and $\Oh(h)$ space, computes for each $i \in [1, h]$ either the maximal $\alpha \in \mathbb Q$ with $\alpha \geq 1$ such that $S_i \cdot \rev{Q}^{\alpha}$ is a prefix of $T$, or reports that such $\alpha$ does not exist.
\end{lemma}

\begin{proof}
For $i \in [1, h]$, let $x_{i} = \absolute{S_i}+1$.
Using constant space and linear time pattern matching (see~\cite{DBLP:journals/tcs/BreslauerGM13}), we enumerate all the occurrences of $\rev{Q}$ in $T$ in increasing or. 
This makes it easy to filter out all the $x_i$ for which $T[x_i\dd ]$ does not have prefix $\rev{Q}$. Hence, from now on we can assume that all the $T[x_i\dd ]$ have prefix $\rev{Q_1}$.
Let $y_i \in [x_i + \absolute{Q}, n]$ be the maximal index such that $T[x_i\dd y_i]$ has period $\absolute{Q}$.
We argue that, if for any $i \in [2, h]$ it holds $x_i \leq y_{i - 1} - \absolute{Q} + 1$, then $y_{i - 1} = y_{i}$. This is easy to see, because then $T[x_{i - 1}\dd y_{i - 1}]$ and $T[x_i\dd y_i]$ have overlap $T[x_i\dd y_{i - 1}]$ of length at least $\absolute{Q}$, such that $T[x_{i - 1}\dd y_{i - 1}]$ and $T[x_i\dd y_i]$ must lie within a single $\absolute{Q}$-periodic fragment.
For every $i \in [1, h]$ in increasing order, we compute $y_i$ as follows. If $i = 1$ or $x_i > y_{i - 1} - \absolute{Q} + 1$, then we naively scan $T[x_i + q\dd ]$ until we reach the end position $y_i$ of the $\absolute{Q}$-periodic fragment (recall that $T[x_i\dd ]$ has prefix $\rev{Q}$ due to the earlier filtering). Otherwise, it holds $x_i \leq y_{i - 1} - \absolute{Q} + 1$, and we can assign $y_i = y_{i - 1}$ in constant time. We scan each position of $T$ at most once (because a scan always starts at a position $x_i + q > y_{i -1}$), and thus the time is $\Oh(n)$, while the space is $\Oh(h)$.
\end{proof}

\subsubsection*{Recursively appending shorter palindromes} 

\begin{restatable}[{Implementation of \cref{lem:new_palindrome_recursive}}]{algo}{newpalrec}\label{algo:new-pal-rec}
    Let $\Aa$ be an affine prefix set of $T[1\dd n]$, given in canonical representation $\rho$ of order $t$.
    Define the set $\Ss$ as
    \[\Ss = \{S \cdot P \mid S \in \Aa, P \in \pal\text{, and $S\cdot P$ is a prefix of $T$} \}.\]
    There is a read-only algorithm that, given $\rho$ and $T$, computes $\Oh((t+1)^2 \log n)$ affine prefix sets whose union is $\Ss$, using $\Oh(n)$ time and $\Oh((t + 1)^3 \log n)$ space. Each affine prefix set is reported in representation of order at most $t + 1$.
\end{restatable}

As explained in \cref{lem:new_palindrome_recursive}, we can build $\Oh((t+1)^2 \log n)$ affine prefix sets such that their union equals $\Ss$ by a recursive application of \cref{algo:inner} and \cref{algo:overhanging}. At every step of the recursion, the order of the affine set decreases by one, and hence the depth of the recursion is $t$. More precisely, the initial step at depth $0$ takes $\Oh(n)$ time and $\Oh((t+1)^2 \log n)$ space by applying \cref{algo:inner} and \cref{algo:overhanging} to the original canonical representation. Then, at depth $i > 0$ of the recursion, we apply the algorithms to an affine prefix set of $Q_{i}^3$, and this set is in canonical representation of order $t - i$. Hence we use $\Oh(\absolute{Q_{i}})$ time and $\Oh((t-i+1)^2 \cdot \log\, \absolute{Q_{i}})$ space at depth $i$. 
This results in $\Oh((t-i+1) \cdot \log\, \absolute{Q_{i}})$ representations of order at most $(t-i+1)$. Then, each of these representations is appended to the same representation of order $i$, resulting in $\Oh((t-i+1) \cdot \log\, \absolute{Q_i})$ representations of order at most $t+1$, taking $\Oh((t + 1) \cdot (t-i+1) \cdot \log\, \absolute{Q_i})$ time and space. Summing over all $i$, the space complexity is
$$\Oh((t+1)^2 \cdot \log n + \sum\nolimits_{i = 1}^{t} (t + 1) \cdot (t-i+1) \cdot \log\, \absolute{Q_i}) = \Oh((t + 1)^3 \log n).$$
The time complexity is the same, with an additional $\Oh(n + \sum_{i = 1}^t \absolute{Q_i})$ needed. Since the representation is canonical, \cref{lem:aux:structureflexflex} implies $\absolute{Q_1} > \sum_{i=2}^t \absolute{Q_i}$ and thus $\Oh(n + \sum_{i = 1}^t \absolute{Q_i}) = \Oh(n)$. Note that $\Oh(n)$ also dominates $\Oh((t + 1)^3 \log n)$ due to \cref{lem:irreducible_maxorder}.

\subsubsection*{Implementation of \cref{th:structure}} 
The final component of the algorithm is a for-loop that goes over all $1 \le i \le k$ to construct a collection $\Cc_i$ of affine sets containing all prefixes of $T$ that belong to $\pal^i$, where each affine set has canonical representation of order at most $i$. 
To construct $\Cc_1$ with $\absolute{\Cc_1} = \Oh(\log n)$, we use \cref{lem:ro-algo-pal1} in $\Oh(n)$ time and $\Oh(\log n)$ space.

Now we inductively show how to compute $\Cc_{i}$ with $i > 1$ from $\Cc_{i - 1}$.
For each affine prefix set in $\Cc_{i - 1}$, or more precisely for its canonical representation (necessarily of order at most $i - 1$), we proceed as follows. We apply \cref{algo:new-pal-rec} to append a palindrome, which takes $\Oh(\absolute{\Cc_{i - 1}} \cdot n)$ time and $\Oh(\absolute{\Cc_{i - 1}} \cdot i^3 \log n)$ space in total, and results in $\Oh(\absolute{\Cc_{i-1}} \cdot i^2 \cdot \log n)$ representations of order at most $i$. 
Next, we have to make the new representations canonical. We apply \cref{algo:makestrong} to each of them, which takes $\Oh(\strongconstplus^{i} \cdot i^2)$ time and $\Oh(\strongconstplus^{i} \cdot i)$ space per representation. The total time is $\Oh(\absolute{\Cc_{i - 1}} \cdot i^4 \cdot \log n \cdot \strongconstplus^{i})$, while the total space is $\Oh(\absolute{\Cc_{i - 1}} \cdot i^3 \cdot \log n \cdot \strongconstplus^{i})$. The result are $\Oh(\absolute{\Cc_{i-1}} \cdot i^2 \cdot \log n \cdot \strongconstplus^{i})$ canonical representations, each of which is of order at most $i$. Finally, $\Cc_{i}$ consists exactly of the sets represented by these representations.

Now we analyze the total time and space complexity. The algorithm terminates after computing $\Cc_k$, where $k$ is the palindromic length of $T$. We focus on a fixed $i \in [2, t]$ and analyze the complexity of computing $\Cc_i$. If we ignore the $\Oh(\absolute{\Cc_{i - 1}} \cdot n)$ time needed to append a palindrome, then the time and space for computing $\Cc_i$, and also its cardinality $\absolute{\Cc_i}$ are bounded from above by some value $U_i = \Oh(\absolute{\Cc_{i - 1}} \cdot i^4 \cdot \log n \cdot \strongconstplus^{i})$. Hence we can find a constant $c \geq 2$ independent of $i$ and $n$ such that $U_i = c \cdot \absolute{\Cc_{i - 1}} \cdot i^4 \cdot \log n \cdot \strongconstplus^{i}$. This also holds for $\Cc_1$ if we define $U_0 = \absolute{\Cc_0} = 1$, and hence we consider $i \in [1, k]$ from now on. We resolve the recurrence and obtain the bound
$$U_i 
= c \cdot \absolute{\Cc_{i - 1}} \cdot i^4 \cdot \log n \cdot \strongconstplus^{i} 
\leq c \cdot U_{i - 1} \cdot i^4 \cdot \log n \cdot \strongconstplus^{i} 
\leq \prod_{j = 1}^i c \cdot j^4 \cdot 6^j \cdot \log n
< c^i \cdot (i!)^4 \cdot 6^{i(i+1)/2} \cdot \log^i n.$$
Let $\smallconst \in \mathbb R^+$ with $\smallconst < 2$ be an arbitrary constant, and let $\smallconst' = \frac1{2-\smallconst} - \frac12 - \frac{1}{2c} > 0$ for $c$ large enough. 
If $i \leq c$ or $(i^i)^6 \geq 6^{\smallconst' \cdot i^2}$, then $i$ is constant and it holds $U_i = \Oh(\log ^i n) = \Oh(6^{i^2/(2-\smallconst)}/i \cdot \log^i n)$. Otherwise, $i > c \geq 2$ implies $i < c^i < i^i$ and thus $c^i \cdot (i!)^4 < (i^i)^6 / i < 6^{\smallconst' \cdot i^2} / i$. We continue with
$$U_i 
< c^i \cdot (i!)^4 \cdot 6^{i(i+1)/2} \cdot \log^i n 
< (6^{\smallconst' \cdot i^2}/i) \cdot 6^{i(i+1)/2} \cdot \log^i n 
= (6^{i^2/(2-\smallconst)}/i) \cdot \log^i n$$
We have shown $U_i = \Oh((6^{i^2/(2-\smallconst)}/i) \cdot \log^i n)$. By summing over all $i \in [1, t]$, we obtain
$$\sum_{i = 1}^k U_i = \Oh(\sum_{i=1}^k (6^{i^2/(2-\smallconst)}/i) \cdot \log^i n) \subseteq \Oh(\sum_{i=1}^k (6^{k^2/(2-\smallconst)}/k) \cdot \log^k n) = \Oh(6^{k^2/(2-\smallconst)} \cdot \log^k n),$$
where the middle step follows from the fact that $(6^{i^2/(2-\smallconst)}/i) \cdot \log^i n \leq (6^{k^2/(2-\smallconst)}/k) \cdot \log^k n$ for arbitrary $n \in \mathbb N^+$ and $i \in [1, k]$.
We have shown that the total space for computing all $\Cc_i$ with $i \in [1, k]$ is $\Oh(6^{k^2/(2-\smallconst)} \cdot \log^k n)$. For the time bound, we ignored the $\Oh(\absolute{\Cc_{i - 1}} \cdot n)$ time needed to compute $\Cc_i$. However, we already know that $\sum_{i = 1}^k \absolute{\Cc_i} < \sum_{i = 1}^k U_i = \Oh(6^{k^2/(2-\smallconst)} \cdot \log^k n)$, and thus the time is bounded by $\Oh(n \cdot 6^{k^2/(2-\smallconst)} \cdot \log^k n)$.
\end{proof}

\subsection{Computing the palindromic length}

The algorithm of \cref{th:ro-algo-palk} can be used to test if the palindromic length of $T$ is at most $k$ by checking whether $T$ is a $k$-palindromic prefix.
We now show how to improve the complexity by using two copies of the data structure of \cref{th:ro-algo-palk}.

\palindromiclength*

\begin{proof}
Let $0 < \smallconst < 1$ be a constant.
	We will start with $k = 1$, and increment $k$ until we can verify that the palindromic length is indeed $k$.
	We will consider strings in $\pal^{\floor{k/2}}$ and $\pal^{\ceil{k/2}}$, and during the complexity analysis we encounter terms of the form $6^{\ceil{k/2}^2/(2-\smallconst)}$.
	Before we describe the algorithm, we point out that these terms can be bounded from above using	
	$$6^{\ceil{k/2}^2/(2-\smallconst)} \leq 6^{((k + 1)/2)^2/(2-\smallconst)} = 6^{(k^2 + 2k + 1)/(8-4\smallconst)} < 6^{(k^2)/(8-4\smallconst)} \cdot 6^{3k}< 6^{k^2/(8-6\smallconst)},$$
	where the last step holds if $k$ exceeds a constant that depends solely on $\smallconst$.
	
    The algorithm calls \cref{th:ro-algo-palk} to build a set $\Pp$ that contains $\Oh(6^{k^2/(8-6\smallconst)} \cdot \log^{\ceil{k/2}} n)$ affine prefix sets that describe the $\ceil{k/2}$-palindromic prefixes of $T$.
    %, where each set is given in canonical representation of order at most $\ceil{k/2}$.
	It then calls \cref{th:ro-algo-palk} again to compute the set $\Ss$ that contains $\Oh(6^{k^2/(8-6\smallconst)} \cdot \log^{\floor{k/2}} n)$ affine prefix sets of $\rev{T}$ that describe the $\floor{k/2}$-palindromic prefixes of $\rev{T}$.
	Each set is given in canonical representation of order at most $\ceil{k/2}$ for $\Pp$, and at most $\floor{k/2}$ for $\Ss$.
	The required time is $\Oh({n \cdot 6^{k^2/(8-6\smallconst)}} \cdot \log^{\ceil{k/2}} n)$, and the overall space is $\Oh(6^{k^2/(8-6\smallconst)} \cdot \log^{\ceil{k/2}} n)$.
	
    Then, we iterate over each possible combination of a set $\Aa \in \Pp$ and a set $\Bb \in \Ss$.
   	We use \cref{lem:verifyprefixsuffix} (below) to check whether there exists $A \in \Aa$ and $B\in\Bb$
    such that $T = A \cdot \rev{B}$.
    Then, $T$ is in $\palk$ if and only if at least one of these checks is succesful.
    There are $\Oh(6^{k^2/(4-3\smallconst)} \cdot \log^{k} n)$ possible combinations, and each combination can be verified in $\Oh(10^k)$ time.
    Hence, the total time needed for verification is $\Oh(10^k \cdot 6^{k^2/(4-3\smallconst)} \cdot \log^{k} n)$.
    
    For the final part of the analysis, recall that we have to increment $k$ and rerun the algorithm until the verification is successful. From now on, let $k$ be the actual palindromic length of $T$. The space is dominated by the final run, which requires $\Oh(6^{k^2/(8-6\smallconst)} \cdot \log^{\ceil{k/2}} n)$ space. The total time is bounded by 
    \begin{alignat*}{3}
    	&\Oh(k\ \cdot\ &&n \cdot 6^{k^2/(8-6\smallconst)} \cdot \log^{\ceil{k/2}} n + k \cdot 10^k\ \cdot\  &&6^{k^2/(4-3\smallconst)} \cdot \log^{k} n)\\
    	\subseteq\ &\Oh(&&n \cdot 6^{k^2/(8-7\smallconst)} \cdot \log^{\ceil{k/2}} n\qquad + &&6^{k^2/(4-4\smallconst)} \cdot \log^{k} n)\\
    	\subseteq\ &\Oh(&&n \cdot \mathrlap{6^{k^2} \cdot \log^{\ceil{k/2}} n}\phantom{6^{k^2/(8-7\smallconst)} \cdot \log^{\ceil{k/2}} n}\qquad + &&6^{k^2} \cdot \log^{k} n),
    \end{alignat*}
    where we used the same trick as before to hide the factors $k$ and $10^k$ by increasing the coefficient of $\smallconst$. 
    If $k \leq \sqrt{\log_{\strongconstplus} n}$, then $\log^k n = o(n)$ and the time complexity is clearly dominated by the term $n \cdot 6^{k^2} \cdot \log^{\ceil{k/2}} n$. 
    If $k > \sqrt{\log_{\strongconstplus} n}$, then $6^{k^2} > n$ and both time and space are superlinear.
    Hence, when running the algorithm, we stop increasing $k$ as soon as it exceeds $\sqrt{\log_{\strongconstplus} n}$.
    If we terminate before, i.e., if the palindromic length is less than $\sqrt{\log_{\strongconstplus} n}$, then we achieve the claimed time complexity.
    Otherwise, i.e., if we terminate because $k$ exceeds $\sqrt{\log_{\strongconstplus} n}$, we finish the computation using the algorithm by Borozdin et al.~\cite{borozdin2017linear}, which takes $\Oh(n) \subseteq \Oh(6^{k^2})$ time and space.    
    %By scaling $\smallconst$, we obtain the complexities claimed in the theorem.
\end{proof}

\begin{lemma}\label{lem:verifyprefixsuffix}
Let $\Aa$ be an affine prefix set of $T$, and let $\Bb$ be affine prefix set $\rev{T}$. If $\Aa$ and $\Bb$ are given in canonical representation of orders respectively at most $t$ and $t'$, then we can decide if there are $A \in \Aa$ and $B \in \Bb$ with $T = A \cdot \rev{B}$ in $\Oh(10^{t + t'})$ time and space.
\end{lemma}

\begin{proof}
Let $\rho = \angles{X, (Q_i, 1, u_i)_{i = 1}^t}$ and $\rho' = \angles{X', (Q_i', 1, u_i')_{i = 1}^{t'}}$ be the respective canonical representations of $\Aa$ and $\Bb$.
We implement the procedure recursively. If $t = t' = 0$, then we can check in constant time if the lengths of the two generated strings sum to $n$.
If $t > 0$ and $t' = 0$ (the case $t = 0$ and $t' > 0$ is symmetric), then $\Bb$ contains a single string of some length $m$, and we only have to check if $\Aa$ contains a string of length $n' = n - m$. 
If $\absolute{XQ_1Q_2Q_3\dots Q_t} > n'$, then every string generated by $\Aa$ is of length over $n'$, and we can terminate with a negative answer.
Otherwise, let 
$$
a_{\min} = {(n' - \absolute{X} - \absolute{Q_2^{u_2}Q_3^{u_3}\dots Q_t^{u_t}}) / \absolute{Q_1}}\text{\quad and \quad}
a_{\max} = {(n' - \absolute{X} - \absolute{Q_2Q_3\dots Q_t}) / \absolute{Q_1}}
$$
(both in $\mathbb Q$), and note that a string $XQ_1^{a_1}Q_2^{a_{2}}\dots Q_t^{a_t}$ with $\forall i \in [1, t] : a_i \in [1, u_i]$ can only be of length $n'$ if $a_{\min} \leq a_1 \leq a_{\max}$. By \cref{lem:aux:structure:stronglength}, it holds $a_{\max} - a_{\min} < 1$, hence there is at most one $a \in \mathbb N$ such that $a_{\min} \leq a \leq a_{\max}$. If $a \in [1, u_1]$, then we replace $\rho$ with $\angles{XQ_1^a, (Q_i, 1, u_i)_{i = 2}^t}$ and recurse. Otherwise, we terminate with a negative answer.

It remains the most general case $t > 0$ and $t' > 0$. If there are $A \in \Aa$ and $B \in \Bb$ such that $T = A \cdot \rev{B}$, then it holds $A=XQ_1^{a_1}Q_2^{a_{2}}\dots Q_t^{a_t}$ and $B={X'Q_1'^{a_1'}Q_2'^{a_{2}'}\dots Q_{t'}^{a_{t'}}}$ for some exponents satisfying $\forall i \in [1, t] : a_i \in [1, u_i]$ and $\forall i \in [1, t'] : a_i' \in [1, u_i']$. For now, assume $\absolute{Q_1} = \absolute{Q_1'}$. Let $a = \min(u_1 - a_1, a_1' - 1)$, then the representations generate strings such that
$$T = XQ_1^{a_1+a}Q_2^{a_{2}}\dots Q_t^{a_t} \enskip \cdot \enskip \rev{X'Q_1'^{a_1'-a}Q_2'^{a_{2}'}\dots Q_{t'}^{a_{t'}}}.$$
Note that either $a_1 + a = u_1$ or $a_1' - a = 1$ (or both). We proceed with two recursive calls. In the first one, we replace $\rho$ with $\angles{XQ_1^{u_1}, (Q_i, 1, u_i)_{i = 2}^t}$. In the second one, we replace $\rho'$ with $\angles{X'Q_1', (Q_i', 1, u_i')_{i = 1}^{t'}}$. If both recursive calls have negative answer, then we terminate with negative answer. Otherwise, we terminate with a positive answer. 

Now we can assume $t > 0$, $t' > 0$, and $\absolute{Q_1} \neq \absolute{Q_1'}$. We only consider $\absolute{Q_1} > \absolute{Q_1'}$, as the other case is symmetric. We again assume that there are $A \in \Aa$ and $B \in \Bb$ such that $T = A \cdot \rev{B}$, where $A$ and $B$ are defined as before.
Since a canonical representation is strongly affine, $A' = XQ_1^{a_1 + 2}Q_2^{a_{2}}\dots Q_t^{a_t}$ is a prefix of $T$.
By \cref{lem:aux:structure:stronglength,lem:aux:structureflexperiod}, it is clear that $A = XQ_1^{a_1}Q_1[1\dd s]$ and $A' = XQ_1^{a_1 + 2}Q_1[1\dd s] = A\cdot (\rot^s(Q_1))^2$ for some $s \in [1, \absolute{Q_1})$. 
If $\absolute{A'} < n - \absolute{X'}$, then the primitive square $(\rot^s(Q_1))^2$ is not only a prefix of $\rev{B} = \rev{X'(Q_1')^{a_1'}(Q_2')^{a_{2}'}\dots (Q_{t'}')^{a_{t'}}}$, but also a prefix of $B' = \rev{(Q_1')^{a_1'}(Q_2')^{a_{2}'}\dots (Q_{t'}')^{a_{t'}}}$. 
However, by \cref{lem:aux:structureflexperiod}, we know that $B'$ and thus also $(\rot^s(Q_1))^2$ has period $\absolute{Q_1'} < \absolute{Q_1}$, which contradicts \cref{lem:primitive_squares}.
Hence we have shown that $\absolute{A'} \geq n - \absolute{X'}$, which implies 
$$a_1 \geq (n - \absolute{X'} - \absolute{X} - \absolute{Q_2^{a_2}Q_3^{a_3}\dots Q_t^{a_t}})/\absolute{Q_1} - 2 > (n - \absolute{X'} - \absolute{X})/\absolute{Q_1} - 3,$$
where the second inequality is due to \cref{lem:aux:structure:stronglength}. We define
$$
a_{\min} = {(n - \absolute{X} - \absolute{X'}) / \absolute{Q_1}} - 3\text{\qquad and \qquad}
a_{\max} = {(n - \absolute{X} - \absolute{X'}) / \absolute{Q_1}}
$$
(both in $\mathbb Q$), and observe that $a_1 < a_{\max}$ because otherwise $\absolute{A} \geq n - \absolute{X'}$ and thus $\absolute{A \cdot \rev{B}} > n$.
We have established $a_{\min} < a_1 < a_{\max}$. It holds $a_{\max} - a_{\min} = 3$, which means that there are at most three possible $a \in \mathbb N$ such that $a_{\min} < a < a_{\max}$. For each $a \in [1, u_1]$ with $a_{\min} < a < a_{\max}$, we recurse by replacing $\rho$ with $\angles{XQ_1^a, (Q_i, 1, u_i)_{i = 2}^t}$. If all of the at most three recursive calls have a negative answer, we terminate with a negative answer. Otherwise, we terminate with a positive answer.

Regardless of the case, we perform at most three recursive calls, and with each call the combined order of the representations decreases by one. Thus, the total number of calls is at most $3^{t+t'}$. In each call, computing the new representations takes $\Oh(t + t')$ time with a naive implementation (only simple arithmetic operations are needed). At all times, the required space is linear in the time spent. Hence the total time and space are $\Oh(3^{t+t'} \cdot (t + t'))$, which is less than $\Oh(10^{(t + t')/2})$ if $t + t'$ exceeds a sufficiently large constant.
\end{proof}

\section{Acknowledgements}
We thank Paweł Gawrychowski for his participation in the initial discussion and, in particular, the idea of the lower bound.

%%%%%%%%%%%%%%%%%%%%%%%%
\bibliographystyle{plainurl}
\bibliography{main-full.bib}

\end{document}